\documentclass{llncs}

\usepackage[utf8]{inputenc}
\usepackage{microtype}%if unwanted, comment out or use option "draft"

\usepackage[version=0.96]{pgf}
\usepackage{tikz}

\usepackage{amsmath}
\usepackage{amssymb}

\usetikzlibrary{positioning}
\usetikzlibrary{arrows,automata}
\usetikzlibrary{calc}
\usetikzlibrary{petri}
\usetikzlibrary{decorations.pathreplacing}
\usetikzlibrary{decorations.pathmorphing}
\usetikzlibrary{trees}
\usetikzlibrary{shapes,shapes.geometric,fit}

\newcommand{\N}{\mathcal{N}}
\newcommand{\W}{\mathcal{W}}
\renewcommand{\L}{\mathcal{L}}
\newcommand{\Nat}{\mathbb{N}}
\newcommand{\F}{\mathcal{F}}

\newcommand {\Ra} {\Rightarrow}
\newcommand {\ra} {\rightarrow}

\newcommand{\Ta}{{\it Ta}}
\newcommand{\by}[2][]{\xrightarrow[#1]{#2}}

\def\mystrut{\rule{0pt}{.6em}}
\newcommand*{\oline}[1]{\overline{#1\mystrut}}

\newtheorem{theoremDummy}{Theorem}
\newenvironment{reftheorem}[1]
{\renewcommand{\thetheoremDummy}{\ref{#1}}\begin{theoremDummy}}
{\end{theoremDummy}}

\newtheorem{lemmaDummy}{Lemma}
\newenvironment{reflemma}[1]
{\renewcommand{\thelemmaDummy}{\ref{#1}}\begin{lemmaDummy}}
{\end{lemmaDummy}}

\makeatletter
\newenvironment{proofsketch}{\par
  \normalfont
  \topsep6\p@\@plus6\p@ \trivlist
  \item[\hskip\labelsep\itshape
    Proof Sketch.]\ignorespaces
}{%
  \endtrivlist
}
\makeatother

\newcommand{\reductionrule}[4]{
\begin{definition}{#1}\\
\label{#2}
\vspace{3pt}
\setlength{\extrarowheight}{5pt}
\begin{tabularx}{\textwidth}{lX}
{\bfseries Guard}: & {
\begin{minipage}[t]{\linewidth}
#3
\end{minipage}}
\\
{\bfseries Action}: &{
\begin{minipage}[t]{\linewidth}
#4
\end{minipage}}
\end{tabularx}

\end{definition}
}

\usepackage{graphicx}
\usepackage{xcolor}

\usepackage{algpseudocode,algorithmicx,algorithm}
\makeatletter
\newcommand{\algmargin}{\the\ALG@thistlm}
\makeatother
\newcommand{\algbox}[1]{\parbox[t]{\dimexpr\linewidth-\algmargin}{#1\strut}}
\usepackage{enumerate}

\usepackage{tabularx}

\usepackage{multirow}

\usepackage{url}
\makeatletter
\g@addto@macro{\UrlBreaks}{\UrlOrds}
\makeatother

\begin{document}
\title{Reduction Rules for Colored Workflow Nets\thanks{This work was partially funded by the DFG Graduiertenkolleg 1480 (PUMA).}}
\author{Javier Esparza \and Philipp Hoffmann}
\institute{Technische Universit\"at M\"unchen}
\maketitle\begin{abstract}
We study Colored Workflow nets \cite{jensen2009coloured}, a model based on Workflow nets \cite{DBLP:journals/jcsc/Aalst98} 
enriched with data. Based on earlier work by Esparza and Desel on the negotiation model of concurrency \cite{negI,negII},
we present reduction rules for our model. Contrary to previous work, 
our rules preserve not only soundness, but also the data flow semantics. 
For free choice nets, the rules reduce all sound nets (and only them) to a net with one single transition and the same 
data flow semantics. We give an explicit algorithm that requires only a polynomial number of 
rule applications.
\end{abstract}

\newlength{\saveold}
\setlength{\saveold}{\textfloatsep}
\section{Introduction}
Workflow Petri nets \cite{DBLP:journals/jcsc/Aalst98,van2004workflow} are 
a very successful formalism for modeling and analyzing
business processes. They have become the most popular formal backend 
for graphical notations like BPMN (Business Process Modeling Notation),
EPC (Event-driven Process Chain), or UML Activity Diagrams, which typically 
do not have a formal semantics. By translating the basic constructs of such 
languages into Petri nets one gets access to a large variety of analysis 
techniques and tools. 

One of these analysis techniques is {\em reduction}.
Reduction algorithms are a very efficient analysis technique for workflows, 
EPCs, AND-XOR graphs and other models (see for instance \cite{sadiq2000analyzing,van2002alternative,van2005verification,verbeek2010reduction}). They consist of a set of {\em reduction rules}, whose application allows one to simplify the workflow while preserving important properties. Reduction aims to elude the state-explosion problem, and, when the property does not hold, 
provides error diagnostics in the form of an irreducible graph \cite{van2002alternative}. Moreover, for certain classes of nets the rules can be {\em complete}, meaning that they reduce all workflows satisfying the property to some 
unique canonical workflow (and only them); in this case, reduction provides a decision algorithm for the property that avoids any kind of state-space exploration. Reduction algorithms are an important part of the well-known 
Woflan tool \cite{DBLP:journals/cj/VerbeekBA01,DBLP:conf/icsoc/OuyangVABDH05}.

Free choice workflow nets (also called workflow graphs) 
are a class of workflow nets that captures many control-flow 
constructs of BPMN, EPC, or Activity Diagrams (see \cite{DBLP:journals/jcsc/Aalst98},
or \cite{DBLP:journals/is/FavreFV15} for a very recent study). In
\cite{van2002alternative} it is shown that a certain set of reduction rules 
for free choice workflow models, originally presented in \cite{desel2005free},
preserves the {\em soundness} property, and is complete. Soundness is 
a fundamental analysis problem for workflows \cite{DBLP:journals/jcsc/Aalst98,van2011soundness}. Loosely speaking, a workflow net is sound if a distinguished marking signaling successful termination is reachable from any reachable marking. The reduction algorithm provides a polynomial-time decision procedure for soundness, in sharp contrast with the fact that 
deciding soundness is at least PSPACE-hard for general workflow nets\footnote{The exact complexity depends on the specifics of the workflow model, for instance whether the workflow Petri net is assumed to be 1-safe or not.}.

However, the rules of \cite{desel2005free} have two important shortcomings. First, while they preserve soundness, they do not preserve any property concerning {\em data}. Workflows manipulating data can be modeled as {\em colored} workflow nets \cite{jensen2009coloured}, 
where tokens carry data values, and transitions transform a tuple of values for its input places into 
a tuple of values for its output places. The {\em linearly dependent place rule} (Rule 2 in Chapter 7 of \cite{desel2005free})
allows one to {\em remove} place $p$ from a net, if it is redundant in the sense that there are other places which together have the same incoming and outgoing transitions as $p$. However, this reduction does 
not make sense for the colored workflow net: the tokens on $p$ might hold a value needed by an outgoing transition 
$t$ to compute the value of the produced tokens! Loosely speaking, the
application of the rule destroys the dataflow semantics of the net. 

\iffalse
To see the problem raised by the rules, consider the 
fragment of a free choice colored workflow net with data shown in Figure \ref{fig:fragment}. 
For example, transition $t_1$ takes a token out of $p_1$, binds the variable $x$ to its value, and 
outputs two tokens to $p_2$ and $p_3$ with values $f(x)$ and $g(x)$. Similarly, transition $t_2$ 
takes tokens from $p_3$ and $p_5$, binds them
to $y$ and $z$, etc. 

\begin{figure}[htbp]
\centering
\input{tikz/introExample}
\caption{A fragment of a workflow net with data}
\label{fig:fragment}
\end{figure}

The {\em linearly dependent place rule} (Rule 2 in Chapter 7 of \cite{desel2005free})
allows one to {\em remove} place $p_2$ from the net, because its row in the incidence matrix is equal to
the sum of the rows of the places $p_3$ and $p_6$ (intuitively, it is the ``sum'' of the places $p_3$ 
and $p_6$, and so in some sense ``redundant''). This is correct with respect to soundness:
the reduced workflow net is sound if{}f the original net is sound. However, the reduction does 
not even make sense for the colored workflow net, because the first value $u$ needed by transition 
$t_3$ to compute the value $i(u,v)$ gets lost! In other words, the
application of the rule totally destroys the semantics of the net. 
\fi

The second shortcoming is that the linearly dependent place rule is not correct for arbitrary 
workflow nets, only for free choice ones (\cite{desel2005free}, page 145\footnote{The example of 
page 145 is not a workflow net, but can be easily transformed into one.}). Since
not all industrial business processes are free choice (30\% of our benchmarks 
in Section \ref{sec:experiments} are non-free choice), this considerably reduces the 
applicability of the rules.  

The most satisfactory solution to these two problems would be to replace the linearly 
dependent place rule by rules extensible to colored nets, while keeping completeness. 
However, this problem has remained open for over 15 years. 
 
In this paper we solve this problem and present a set of surprisingly simple 
rules that overcomes the shortcomings. First, 
the rules can be applied to arbitrary colored workflow nets. Second, they preserve not 
only the sound/unsound character of the net, but also the {\em input/output relation} of 
the workflow; more precisely, the original workflow net has a firing sequence that transforms an entry token
with value $v_{in}$ into an exit token with value $v_{out}$ if{}f the net after the reduction
also has such a sequence. Therefore, the rules can be applied to decide any property
of the input/output relation. Finally, the new rules are complete for free choice 
workflow nets.

Our results rely on previous work on {\em negotiations}, a model of concurrency
introduced in \cite{negI,negII}. Negotiations share many features with Petri nets,
but, unlike Petri nets, are a structured model of communicating sequential agents. 
In \cite{negII} a complete set of 
reduction rules for the class of {\em deterministic negotiations} is presented. We 
generalize the results of \cite{negII} to show that a similar set of rules is correct 
for arbitrary workflow nets, and complete for free choice workflow nets. Since the proofs 
of \cite{negII} make strong use of the agent structure,
we must substantially modify them, and in fact write many of them from scratch. 
Moreover, because of the agent structure of negotiations, 
workflow nets obtained as translations of negotiations are automatically 1-safe. Therefore,
the results cannot be used to deal with variants of the soundness notion, like
$k$-soundness or generalized soundness \cite{van2011soundness}. Making use of the theory of 
free choice nets we can however show that our rules are still correct and complete for these variants. 

Finally, and as a third contribution of the paper, we report on
some experimental results. In \cite{negII} only the rules and the 
completeness result are presented, but neither a specific algorithm prescribing a concrete strategy to
decide which rule to apply at which point, nor an implementation and experimental validation.
In this paper we report on a prototype implementation, and on experimental results on a benchmark suite 
of nearly 2000 workflows derived from industrial business processes. 

\paragraph{Other related work.} The soundness problem has been extensively studied, both from a theoretical
and a practical point of view, and very efficient verification algorithms have been developed 
(see e.g. \cite{van2011soundness} for a comprehensive survey). Our approach is not more efficient for 
checking soundness than the ones of e.g. \cite{fahland2009instantaneous}, but can also be applied to checking arbitrary properties of the input/output relation, while retaining completeness. In \cite{sadiq2004data,trvcka2009data} state-space exploration of workflows is performed to identify data flow anti-patterns
(like a variable being assigned a value during an execution, but never being read afterwards).
Our technique aims at avoiding state-space exploration and considers properties of the input/output relation.

The paper is organized as follows. Section \ref{sec:workflowNets} defines
workflow nets, free choice nets, and soundness. Section \ref{sec:rules} presents 
our reduction rules and proves them correct. In Section \ref{sec:procedure} we first show completeness 
for acyclic nets and then extend the result to cyclic nets. Section \ref{sec:experiments} 
presents experimental results on the benchmarks of \cite{van2007verification,fahland2009instantaneous}. 
Finally, Section \ref{sec:conclusion} contains some conclusions 
and open questions. The proofs of all results can be found in the appendix.

\section{Workflow Nets and Colored Workflow Nets}
\label{sec:workflowNets}

We recall the definitions of workflow nets and the soundness property.

\begin{definition}[Workflow net]\cite{DBLP:journals/jcsc/Aalst98}
A Workflow net (WF net) is a quintuple $(P,T,F,i,o)$ where
\begin{itemize}
\item $P$ is a finite set of places.
\item $T$ is a finite set of transitions ($P\cap T = \emptyset$).
\item $F \subseteq (P\times T) \cup (T \times P)$ is a set of arcs.
\item $i, o \in P$ are places such that $i$ has no incoming arcs, $o$ has no outgoing arcs.
\item The graph $(P\cup T, F\cup(o, i))$ is strongly connected.
\end{itemize}
\end{definition}

We write ${}^\bullet p$ and $p^\bullet$ to denote the input and output
transitions of a place $p$, respectively, and similarly ${}^\bullet t$ and 
$t^\bullet$ for the input and output places of a transition $t$.
A marking $M$ is a function from $P$ to the natural numbers that assigns a 
number of tokens to each place. A transition $t$ is enabled 
at $M$ if all places of ${}^\bullet t$ contain at least one token in $M$. 
An enabled transition may fire, removing a token 
from each place of ${}^\bullet t$ and adding one token to each 
place of $t^\bullet$. The {\em initial marking} ({\em final marking}) 
of a workflow net puts one token on place $i$ (on place $o$), 
and no tokens elsewhere. A marking is {\em reachable} if some sequence of
transition firings leads from the initial marking to it. We call elements in $P\cup T$ the nodes of the workflow net.

\begin{definition}[Soundness]\cite{DBLP:journals/jcsc/Aalst98}
A WF net $\W=(P,T,F,i,o)$ is sound if
\begin{itemize}
\item the final marking is reachable from any reachable marking, and 
\item every transition occurs in some firing sequence starting from the initial marking.
\end{itemize}
\end{definition}

When modeling a workflow, it is useful to model not only control flow 
but also data flow. We do so by means of Colored Workflow nets.

\begin{definition}[Colored WF net]\cite{jensen2009coloured}
A {\em colored WF net} (CWF net) is a tuple $\W=(P,T,F,i,o,V,\lambda)$ where 
$(P,T,F,i,o)$ is a WF net, $V$ is a function that assigns to 
every place $p \in P$ a \emph{color set} $C_p$ and $\lambda$ is a function that
assigns to each transition $t\in T$ 
a left-total relation $\lambda(t)\subseteq\prod_{p\in{}^\bullet t} C_p \times \prod_{p\in t\bullet} C_p$ 
between the values of the input places and those of the output places of $t$. 

A \emph{colored marking} $M$ of $\W$ is a function that assigns to each place 
$p$ a multiset $M(p)$ over $C_p$, interpreted as a multiset of {\em colored tokens} currently on $p$. A colored marking is {\em initial} ({\it final}) if it puts one token on place $i$ (on place $o$), of any color in $C_i$ ($C_o$), and no tokens elsewhere. 
\end{definition}

Observe that there are as many initial markings as elements in $C_i$. To distinguish between input and output values of a transformer $\lambda$, 
we separate them by a $\ra$.

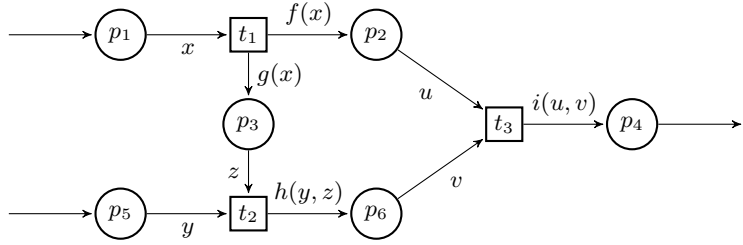
\begin{figure}[tb]
\centering
\begin{tikzpicture}[>=stealth',bend angle=45,auto,scale=1.7]
  \tikzstyle{place}=[circle,thick,draw=black,fill=white,minimum size=6mm]
  \tikzstyle{wt}=[rectangle,thick,draw=white,fill=white,minimum size=4mm]
  \tikzstyle{transition}=[rectangle,thick,draw=black,fill=white,minimum size=4mm]
  \tikzstyle{every label}=[black]
   
    \node [place] at (0,0) (p1)                 {$p_1$};
    \node [place] at (2,0) (p2)                 {$p_2$};
    \node [place] at (1,-0.7) (p3)                 {$p_3$};
    \node [place] at (4,-0.7) (p4)                 {$p_4$};  
    \node [place] at (0,-1.4) (p5)                 {$p_5$};
    \node [place] at (2,-1.4) (p6)                 {$p_6$};

    \node [transition] at (1,0)  (t1) {$t_1$}
      edge [pre]  node{$x$}      (p1)
      edge [post] node{$f(x)$} (p2)
      edge [post] node{$g(x)$} (p3);

    \node [transition] at (1,-1.4)  (t2) {$t_2$}
      edge [pre]  node {$z$}      (p3)
      edge [pre]  node {$y$}      (p5)
      edge [post] node {$h(y,z)$} (p6);

    \node [transition] at (3,-0.7)  (t3) {$t_3$}
      edge [pre]  node {$u$}      (p2)
      edge [pre]  node {$v$}      (p6)
      edge [post] node {$i(u,v)$} (p4);
  
    \node [wt] at (-1,0) (t0d) {}
      edge [post]                (p1);

    \node [wt] at (-1,-1.4) (t1d) {}
      edge [post]                (p5);

    \node [wt] at (5,-0.7) (t2d) {}
      edge [pre]                 (p4);

\end{tikzpicture}
\caption{A partial workflow net with data}
\label{fig:fragment}
\end{figure}

Consider the partial workflow net in Figure \ref{fig:fragment} and take $C_p = \Nat$ for every place $p$ of the net.
An example of a colored marking could be the marking $( \{3\}, \emptyset, \emptyset, \emptyset, \{2,4\}, \emptyset)$ which puts a token of color 3 on $p_1$ and two tokens, one of color 2 and one of color 4, on $p_5$.
If $f(x) = x+1$ and $g(x) = x+2$, then we have 
$\lambda(t_1) = \{ (n \ra n+1, n+2) \mid n \geq 0 \}$. 

We call $\lambda(t)$ the transformer associated with $t$. 
When a transition $t$ fires, the colored marking changes in the 
expected way \cite{jensen2009coloured}: (a) remove a token from 
each input place of $t$; (b) choose an element of $\lambda(t)$ whose 
projection onto the input places matches the tuple of removed tokens; 
(c) add the projection of $\lambda(t)$ onto the output places 
to the output places of $t$. We write $M \by{t}M'$ to denote that $t$ is 
enabled at $M$ and its firing leads to $M'$. For example, the colored marking
$( \{3\}, \emptyset, \emptyset, \emptyset, \{2,4\}, \emptyset)$ enables 
transition $t_1$, and taking $h(y,z) = y \cdot z$ we have 
$$( \{3\}, \emptyset, \emptyset, \emptyset, \{2,4\}, \emptyset)
\by{t_1} (\emptyset, \{4\}, \{5\}, \emptyset, \{2,4\}, \emptyset)
\by{t_2} (\emptyset, \{4\}, \emptyset, \emptyset, \{4\}, \{10\}) \ .$$

\subsection{A colored version of the insurance claim example}
\label{sec:example}

We extend the well known insurance complaint process of \cite{DBLP:journals/jcsc/Aalst98} with data. 
The workflow is shown in Figure \ref{fig:claim}. 
After initial registration of the complaint, 
a questionnaire is sent to the complainant. 
In parallel, the complaint is evaluated. The evaluation decides whether processing is required. In that case, the processing takes place (e.g. by some employee) and is checked for correctness (e.g. by a senior employee) which may either lead to another round of processing if an error is found, or the processing ends. Finally, the complaint is archived.

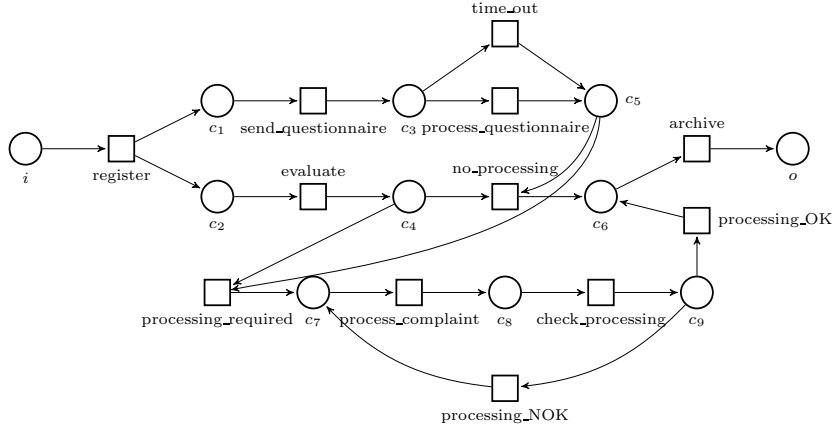
\begin{figure}[t]
\centering
\scalebox{0.85}{
\begin{tikzpicture}[>=stealth',bend angle=45,auto,scale=1.5]
	\tikzstyle{every node}=[font=\scriptsize]
	\tikzstyle{place}=[circle,thick,draw=black,fill=white,minimum size=5mm,inner sep=0mm]
	\tikzstyle{transition}=[rectangle,thick,draw=black,fill=white,minimum size=4mm,inner sep=0mm]
	\tikzstyle{every label}=[black]
	
	\node [place] at (0,0) (i){};
	\node [below=0cm of i] {$i$};
	\node [place] at (2,0.5) (c1){};
	\node [below=0cm of c1] {$c_1$};
	\node [place] at (2,-0.5) (c2){};
	\node [below=0cm of c2] {$c_2$};
	\node [place] at (4,0.5) (c3){};
	\node [below=0cm of c3] {$c_3$}; 
	\node [place] at (4,-0.5) (c4){};
	\node [below=0cm of c4] {$c_4$};
	\node [place] at (6,0.5) (c5){};
	\node [right=0cm of c5] {$c_5$};
	\node [place] at (6,-0.5) (c6){};
	\node [below=0cm of c6] {$c_6$};
	\node [place] at (3,-1.5) (c7){};
	\node [below=0cm of c7] {$c_7$};
	\node [place] at (5,-1.5) (c8){};
	\node [below=0cm of c8] {$c_8$};
	\node [place] at (7,-1.5) (c9){};
	\node [below=0cm of c9] {$c_9$};
	\node [place] at (8,0) (o){};
	\node [below=0cm of o] {$o$};

	\node [transition] at (1,0) (t1) {}
	edge [pre] (i)
	edge [post] (c1)
	edge [post] (c2);
	\node [below=0cm of t1] {register};

	\node [transition] at (3,0.5) (t2) {}
	edge [pre] (c1)
	edge [post] (c3);
	\node [below=0cm of t2] {send\_questionnaire};

	\node [transition] at (5,0.5) (t3) {}
	edge [pre] (c3)
	edge [post] (c5);
	\node [below=0cm of t3] {process\_questionnaire};

	\node [transition] at (5,1.2) (t4) {}
	edge [pre] (c3)
	edge [post] (c5);
	\node [above=0cm of t4] {time\_out};

	\node [transition] at (3,-0.5) (t5) {}
	edge [pre] (c2)
	edge [post] (c4);
	\node [above=0cm of t5] {evaluate};

	\node [transition] at (5,-0.5) (t6) {}
	edge [pre] (c4)
	edge [pre,bend right=30] (c5)
	edge [post] (c6);
	\node [above=0cm of t6] {no\_processing};

	\node [transition] at (2,-1.5) (t7) {}
	edge [pre] (c4)
	edge [post] (c7);
	\node [below=0cm of t7] {processing\_required};

	\draw[->] (c5) .. controls (5.9, -1) and (3, -1.3).. (t7);

	\node [transition] at (4,-1.5) (t8) {}
	edge [pre] (c7)
	edge [post] (c8);
	\node [below=0cm of t8] {process\_complaint};

	\node [transition] at (6,-1.5) (t9) {}
	edge [pre] (c8)
	edge [post] (c9);
	\node [below=0cm of t9] {check\_processing};

	\node [transition] at (5,-2.5) (t10) {}
	edge [pre,bend right=20] (c9)
	edge [post,bend left=20] (c7);
	\node [below=0cm of t10] {processing\_NOK};

	\node [transition] at (7,-0.75) (t11) {}
	edge [pre] (c9)
	edge [post] (c6);
	\node [right=0cm of t11] {processing\_OK};

	\node [transition] at (7,0) (t12) {}
	edge [pre] (c6)
	edge [post] (o);
	\node [above=0cm of t12] {archive};

\end{tikzpicture}
}
\caption{Insurance claim process}
\label{fig:claim}
\end{figure}

We add colors to keep track of the status of the complaint and its estimated cost for the company, modeled by a number in the interval $[1..10]$ (see Table \ref{tab:datavalues}). Furthermore each claimant belongs to a customer group, either \textrm{A} or \textrm{B}. \textrm{A}'s and \textrm{B}'s insurance policies entitle them, respectively, to the full cost or to half the cost of the damage. 
The color sets of places $i,o,c_2,c_6$ are the pairs 
$\{A, B\} \times [1..10]$, modeling the customer group and the cost of the claim as estimated by the customer. The colors of place
$c_4$ additionally contain the result of the evaluation: PR (process) or NPR (do not process). Colors of $c_5$ store the result of the questionnaire:
the answer to the question ``was it your fault?'' (YES/NO), or a time out (TO). In place $c_7$, the information from $c_4$ and $c_5$ is put together, and in $c_8$ the result of the first processing is added. Finally, tokens in
$c_9$ can have the same values as those in $c_8$, plus an additional value ERR if the check at transition \texttt{check\_processing} reveals a miscalculation. Tokens in $c_6$ and $o$ store the amount that was actually paid by the company after the processing was successful (or without processing).

\begin{table}[t]
\begin{center}
\begin{tabular}{lll}
$C_i = C_o = C_{c_2} = C_{c_6} = \{A,B\} \times [1..10]$ & \qquad \qquad & $C_{c_7} = C_{i} \times C_{c_5}$ \\
$C_{c_1} = C_{c_3} = \{\bullet\}$ &  & $C_{c_8} = C_{c_7} \times [1..10]$ \\
$C_{c_4} = C_i \times \{\rm PR,NPR\} $ & & $C_{c_9} = C_{c_7} \times ([1..10] \cup \{{\rm ERR}\})$\\
$C_{c_5} = \{{\rm YES}, {\rm NO}, {\rm TO}\}$ & &  \\
\end{tabular}
\end{center}

\begin{center}
$\begin{array}{rcl}
\lambda(\texttt{register}) & = & \{(x,k \ra \{\bullet\}\times\{x,k\})\mid 1\leq k\leq 10 \} \\
\lambda(\texttt{send\_questionnaire})&=&\{(\bullet \ra \bullet)\}\\
\lambda(\texttt{time\_out}) & = & \{(\bullet\ra {\rm TO})\}\\
\lambda(\texttt{process\_questionnaire}) & = & \{(\bullet\ra {\rm YES}),(\bullet\ra {\rm NO})\}\\
\lambda(\texttt{evaluate}) & = & \{(x,k\ra x,k,{\rm NPR})\mid 1\leq k\leq 3\}\\
&\cup & \{(x,k\ra x,k,{\rm PR})\mid 4\leq k\leq 10\}\\
\lambda(\texttt{no\_processing}) & = & \{(x,k,{\rm NPR},q\ra x,k) \mid 1\leq k\leq 3\}\\
\lambda(\texttt{processing\_required}) & = & \{(x,k,{\rm PR},q\ra x,k,q) \mid 4\leq k\leq 10\}\\
\lambda(\texttt{process\_complaint}) & = & \{(x,k,q\ra x,k,q,v) \mid 4\leq k\leq 10, 1\leq v \leq k\}\\
\lambda(\texttt{check\_processing})& = & \{(x,k,v,q\ra x,k,q,v) \mid x = {\rm A}, 4\leq k\leq 10, v = k\}\\
& \cup & \{(x,k,v,q\ra x,k,q,v) \mid x = {\rm B}, 4\leq k\leq 10, v = k/2\}\\
& \cup & \{(x,k,v,q\ra x,k,q,ERR) \mid \text{otherwise}\}\\
\lambda(\texttt{processing\_NOK}) & = & \{(x,k,q,{\rm ERR}\ra x,k,q)  \mid 4\leq k\leq 10\}\\
\lambda(\texttt{processing\_OK}) & = & \{(x,k,q,v\ra x,v) \mid 4\leq k\leq 10, 1\leq v\leq 10\}\\
\lambda(\texttt{archive}) & = & \{(x,v\ra x,v) \mid (x,v) \in C_{c_6}\}
\end{array}$
\end{center}
\caption{Color sets and transformers for the insurance claim workflow}
\label{tab:datavalues}
\end{table}

Assume that the company's policy is to accept all claims which are evaluated to a value of 3 or less without any further processing, and process all other claims. The transformers modeling this policy are given in Table \ref{tab:datavalues}, where $x\in\{{\rm A,B}\}$ and $q\in\{{\rm YES, NO, TO}\}$ unless otherwise stated. Division by 2 is assumed to be integer division.

All transformers are self-explanatory except perhaps 
\texttt{process\_complaint} and  \texttt{check\_processing}.
 In \texttt{process\_complaint}, an employee may lower the customer's estimate $k$ to a new value $v$. In \texttt{check\_processing}, a senior employee checks that the employee made no mistake (modeled by the fact that $v$ must be $k/2$ or $k$ depending on the customer group). If the check fails, an error flag is set and the processing is repeated. 

Apart from the soundness of the workflow, we 
wish to check the following property: if two customers in the 
same group register insurance complaints, then the one claiming a higher 
also receives a higher amount (notice that our ideal insurance
company does not reject any complaint). We shall use our 
reduction algorithm to check that the property holds for customers of
group A, but not for customers of group {\rm B}. 

The attentive reader may have noticed that the semantics of colored nets
allows, e.g., to take the transition \texttt{no\_processing}
even when the evaluation indicates that processing is necessary. 
This can easily be dealt with by introducing additional error values 
that are then propagated until the end. We omit them to ease the reading 
and assume that \texttt{no\_processing} and \texttt{processing} 
are taken according to the result of \texttt{evaluate}, and similarly in other cases.

\iffalse
\begin{figure}
\begin{center}
\input{tikz/travel}
\end{center}
\caption{Travel agency workflow}
\label{fig:claim}
\end{figure}

For another example, consider the workflow in a travel agency depicted in Figure : a customer can request travel information. The request is registered and it is checked whether the request comes from a new or an established customer. Thereafter an employee of the travel agency checks for opportunities and presents them to the customer. The customer can either not be interested, can ask for more options or decide to book a trip. There are two booking choices: with or without insurance. After the trip is booked, the customer can take the trip or canceling it. In the second case, depending on his choice regarding insurance, he will get some money back.

The data we will use includes a flag whether it is a new or established customer, the amount the customer has paid already and whether or not he chose to book an insurance.
\fi

\subsection{Summaries and Equivalence}
\label{subsec:summ}

Since a workflow net describes a process starting at  
$i$ and ending at $o$, it is interesting to study the input/output 
relation or {\em summary} of the whole process.

\begin{definition}[Summary and equivalence]
Let $\W$ be a colored WF net. Let $\mathcal{M}_i$ and $\mathcal{M}_o$ be the sets of initial and final colored markings of $\W$. The {\em summary} of $\W$ is the relation $S\subseteq \mathcal{M}_i \times \mathcal{M}_o$ given by: $(M_i,M_o) \in S$ if{}f $M_o$ is reachable from $M_i$. Two colored WF nets are {\em equivalent} if{}f they are both sound or both unsound, and have the same summary.
\end{definition}

Our rules aim to reduce CWF nets while preserving equivalence. 
If we are able to reduce a CWF to another one with one single transition $t$, 
then the summary is given by $\lambda(t)$, and we say that the CWF has been
{\em completely reduced} and we have {\em computed the summary}. 
Since this CWF net is obviously sound and rules preserve equivalence, 
if a CWF net can be completely reduced, then it is sound.
We prove that our rules preserve equivalence for all CWF nets, 
and give an algorithm that completely reduces all 
sound {\em free choice CWF nets}, defined below, by means of a polynomial number of rule applications.

In Section \ref{sec:procedure} we compute the summary of the free choice CWF net of Figure \ref{fig:claim}
using our reduction procedure. The result (where we write $M_i \Ra M_o$ instead of $(M_i,M_o) \in S$, and omit the error values) is:

\begin{equation*}
\begin{array}{rcl}

\{ (A,k \Ra A,k) \mid 1 \leq k \leq 10 \} & \cup & \{(B,k \Ra B,k) \mid 1\leq k\leq 3 \}  \\
 & \cup  & \{(B,k \Ra B,k/2) \mid 4\leq k\leq 10 \}
\end{array}
\end{equation*}

Since the summary  contains $(B,3 \Ra B,3)$ and $(B, 4 \Ra B, 2)$, the company policy does not satisfy 
the desired property for customers of group {\rm B}.

\subsection{Free choice Workflow Nets}

We recall the definition of free choice workflow nets \cite{desel2005free,DBLP:journals/jcsc/Aalst98}.

\begin{definition}[Free choice workflow nets]
A workflow net $\W=(P,T,F,i,o)$ is {\em free choice} (FC) if for every two places $p_1, p_2 \in P$
either $p_1^\bullet \cap p_2^\bullet = \emptyset$ or $p_1^\bullet = p_2^\bullet$.
\end{definition}

The net of Figure \ref{fig:claim} is free choice.
We also need to introduce clusters, and the new notion of free choice cluster and free choice node.

\begin{definition}[Clusters, free choice nodes] \cite{desel2005free}
Let $\W=(P,T,F,i,o)$ be a workflow net. The {\em cluster} of $x \in P \cup T$ is the unique 
smallest set $[x] \subseteq P \cup T$ satisfying: $x \in [x]$, if $p \in P \cap [x]$ then $p^\bullet \subseteq [x]$, and if $t \in T \cap [x]$, then ${}^\bullet t \subseteq [x]$. A set $X \subseteq P \cup T$ is a cluster if $X = [x]$ for some $x$. A cluster $c$ is {\em free choice} if 
$(p, t) \in F$ for every $p \in P \cap c$ and $t \in T \cap c$. 
A node $x$ is {\em free choice} if $[x]$ is a free choice cluster.
\end{definition}

The sets $\{c_3 \} \cup c_3^\bullet$ and $\{c_4, c_5\} \cup c_4^\bullet \cup c_5^\bullet$ are free choice clusters of the net of Figure \ref{fig:claim}. 
It is easy to see that clusters are equal or disjoint, and 
therefore the clusters of $\W$ are a partition of $P \cup T$. Further, we have  $[i] \cap P = \{i\}$ and $[o] = \{o\}$. Finally, we have that $\W$ is free choice if{}f all its nodes are free choice.

We say that a marking $M$ marks a cluster $c$ if it marks {\em all} places in $c$. 
Observe that if a cluster is marked, then all its transitions are enabled.  
We say that a cluster fires if one of its transitions fires.

\section{Reduction rules}
\label{sec:rules}

We present a set of three \emph{reduction rules} for CWF nets similar to those 
used for transforming finite automata into regular expressions \cite{Hopcroft:2006:IAT:1196416}.

A reduction rule, or just rule, is a binary relation on the set of CWF nets. For a rule $R$, 
we write $W_1 \by{R} W_2$ for $(W_1,W_2) \in R$. A rule $R$ is \emph{correct}
if it preserves equivalence, i.e., if $W_1 \by{R} W_2$ implies that $W_1$ and $W_2$ are equivalent. 

Given a set of rules $\mathcal{R}=\{R_1,\ldots,R_k\}$, we denote by $\mathcal{R}^\ast$ the transitive 
closure of $R_1\cup\ldots\cup R_k$. We say that $\mathcal{R}$ is \emph{complete} for a class of CWF 
nets if for every sound CWF net $\W$ in that class there is a CFW net $\W'$ consisting of a single 
transition between the two only places $i$ and $o$ such that $\W \by{\mathcal{R}^\ast} \W'$. 

We describe rules as pairs of a \emph{guard} and an \emph{action}. 
$W_1 \by{R} W_2$ holds if $W_1$ satisfies the guard, and $W_2$ is a 
possible result of applying the action to $W_1$. 

\paragraph*{Merge rule.} Intuitively, the {\em merge rule} merges two transitions with 
the same input and output places into one single transition. 

\reductionrule{Merge rule}{def:merge}
{
$\W$ contains two distinct transitions $t_1, t_2 \in T$ such that ${}^\bullet t_1 = {}^\bullet t_2$ and $t_1^\bullet = t_2^\bullet$.
}{
\vspace{-\topsep}
\begin{enumerate}[(1)]
\item $T := (T\setminus \{t_1,t_2\})\cup \{t_m\}$, where $t_m$ is a fresh name. 
\item $t_m^\bullet := t_1^\bullet$ and ${}^\bullet t_m := {}^\bullet t_1$. 
\item $\lambda (t_m) := \lambda (t_1) \cup \lambda (t_2)$.
\end{enumerate}
}

\paragraph*{Iteration rule.} Loosely speaking, the iteration rule replaces arbitrary iterations of a transition 
by a single transition with the same effect.

\reductionrule{Iteration rule}{def:iteration}
{
$\W$ contains a free choice cluster $c$ with a transition $t \in c$ such that $t^\bullet={}^\bullet t$. 
}{
\vspace{-\topsep}
\begin{enumerate}[(1)]
\item $T := (T \setminus \{t\})$. 
\item For all $t' \in c \setminus\{t\}$: $\lambda(t'):=\lambda(t)^\ast\cdot\lambda(t')$ where $\lambda(t)^\ast= \sum_{i\geq 0} \lambda(t)^i$, and $\lambda(t)^0$ is the identity relation.
\end{enumerate}
}

Observe that $\lambda(t)^\ast$ captures the fact that $t$ can be executed arbitrarily often.

\paragraph*{Shortcut rule.} The shortcut rule merges transitions of two clusters, one of which will occur 
as a consequence of the other, into one single transition with the same effect. 

\begin{definition}
A transition $t$ {\em unconditionally enables} a cluster $c$
if $c \cap P \subseteq t^\bullet$.
\end{definition}

\noindent Observe that if $t$ unconditionally enables $c$ and a marking $M$ enables $t$,
then the marking $M'$ given by $M \by{t} M'$ enables every transition in $c$.

\reductionrule{Shortcut rule}{def:shortcut}
{
$\W$ contains a transition $t$ and a free choice cluster $c\notin\{[o],[t]\}$ such that $t$ unconditionally enables $c$.
}{
\vspace{-\topsep}
\begin{enumerate}[(1)]
\item $T := (T \setminus \{t\}) \cup \{t'_s \mid t' \in c\}$, where $t'_s$ are fresh names. 
\item For all $t' \in c$: ${}^\bullet t'_s := {}^\bullet t$ and $t'_s{}^\bullet := (t^\bullet \setminus {}^\bullet t')\cup t'^\bullet$. 
\item For all $t' \in c$: $\lambda (t'_s) := \lambda (t)\cdot\lambda(t')$. 
\item If ${}^\bullet p = \emptyset$ for all $p\in c$, then remove $c$ from $\W$.
\end{enumerate}
}

We also use a restricted version of this rule, called the {\em d-shortcut rule}. This rule is obtained by adding an additional guard to the shortcut rule: $|c \cap T|=1$. This guard guarantees that the number of edges does not increase when the d-shortcut rule is applied.

Figure \ref{fig:rules} shows a sequence of reductions illustrating the 
definitions of the rules. Notice that the graphical description does 
not contain the transformer information. A second example of reduction
in which the workflow net also exhibits concurrency is shown in Section
\ref{subsec:summexample}.

\begin{theorem}
\label{thm:correctness}
The merge, shortcut and iteration rules are correct for CWF nets.
\end{theorem}

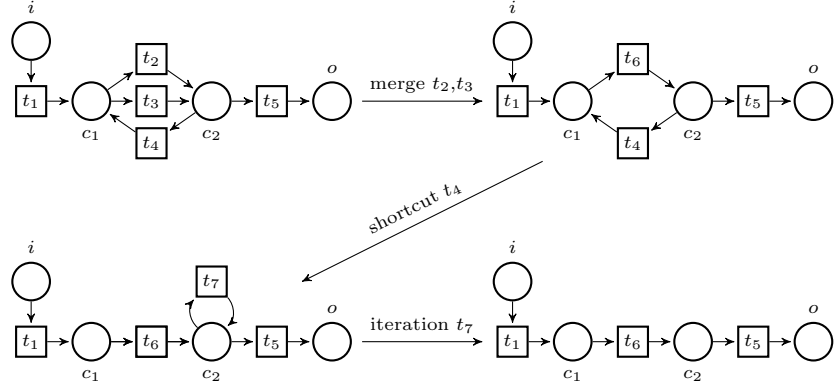
\begin{figure}
\centering
\begin{tikzpicture}[>=stealth',bend angle=45,auto,scale=0.8]
	\tikzstyle{every node}=[font=\scriptsize]
	\tikzstyle{place}=[circle,thick,draw=black,fill=white,minimum size=5mm,inner sep=0mm]
	\tikzstyle{transition}=[rectangle,thick,draw=black,fill=white,minimum size=4mm,inner sep=0mm]
	\tikzstyle{every label}=[black]
	
	\node [place] at (-1,-1) (c0){};
	\node [above=0cm of c0] {$i$};
	\node [place] at (0,-2) (c1){};
	\node [below=0cm of c1] {$c_1$};
	\node [place] at (2,-2) (c2){};
	\node [below=0cm of c2] {$c_2$};
	\node [place] at (4,-2) (c3){};
	\node [above=0cm of c3] {$o$};

	\node [transition] at (-1,-2) (t1) {$t_1$}
	edge [pre] (c0)
	edge [post] (c1);

	\node [transition] at (1,-1.3) (t2) {$t_2$}
	edge [pre] (c1)
	edge [post] (c2);

	\node [transition] at (1,-2) (t3) {$t_3$}
	edge [pre] (c1)
	edge [post] (c2);

	\node [transition] at (1,-2.7) (t4) {$t_4$}
	edge [pre] (c2)
	edge [post] (c1);	

	\node [transition] at (3,-2) (t5) {$t_5$}
	edge [pre] (c2)
	edge [post] (c3);

\draw[->] (4.5,-2) -- (6.5,-2) node [midway] {merge $t_2$,$t_3$};

\begin{scope}[shift={(8,0)}]
	\node [place] at (-1,-1) (c0){};
	\node [above=0cm of c0] {$i$};
	\node [place] at (0,-2) (c1){};
	\node [below=0cm of c1] {$c_1$};
	\node [place] at (2,-2) (c2){};
	\node [below=0cm of c2] {$c_2$};
	\node [place] at (4,-2) (c3){};
	\node [above=0cm of c3] {$o$};

	\node [transition] at (-1,-2) (t1) {$t_1$}
	edge [pre] (c0)
	edge [post] (c1);

	\node [transition] at (1,-1.3) (t2) {$t_6$}
	edge [pre] (c1)
	edge [post] (c2);

	\node [transition] at (1,-2.7) (t4) {$t_4$}
	edge [pre] (c2)
	edge [post] (c1);	

	\node [transition] at (3,-2) (t5) {$t_5$}
	edge [pre] (c2)
	edge [post] (c3);
\end{scope}

\draw[->] (7.5,-3) -- (3.5,-5) node [midway,sloped,swap,anchor=center,above] {shortcut $t_4$};

\begin{scope}[shift={(0,-4)}]
	\node [place] at (-1,-1) (c0){};
	\node [above=0cm of c0] {$i$};
	\node [place] at (0,-2) (c1){};
	\node [below=0cm of c1] {$c_1$};
	\node [place] at (2,-2) (c2){};
	\node [below=0cm of c2] {$c_2$};
	\node [place] at (4,-2) (c3){};
	\node [above=0cm of c3] {$o$};

	\node [transition] at (-1,-2) (t1) {$t_1$}
	edge [pre] (c0)
	edge [post] (c1);

	\node [transition] at (1,-2) (t2) {$t_6$}
	edge [pre] (c1)
	edge [post] (c2);

	\node [transition] at (1,-2) (t3) {$t_6$}
	edge [pre] (c1)
	edge [post] (c2);

	\node [transition] at (2,-1) (t4) {$t_7$}
	edge [pre,bend right] (c2)
	edge [post, bend left] (c2);	

	\node [transition] at (3,-2) (t5) {$t_5$}
	edge [pre] (c2)
	edge [post] (c3);
\end{scope}
% \draw[->] (12.5,-1) -- (14.5,-1) node [midway] {iteration $t_7$};
% \begin{scope}[shift={(15,0)}]
% 	\node [place] at (0,0) (c0){};
% 	\node [above=0cm of c0] {$i$};
% 	\node [place] at (0,-2) (c1){};
% 	\node [below=0cm of c1] {$c_1$};
% 	\node [place] at (2,-2) (c2){};
% 	\node [below=0cm of c2] {$c_2$};
% 	\node [place] at (2,0) (c3){};
% 	\node [above=0cm of c3] {$o$};

% 	\node [transition] at (0,-1) (t1) {$t_1$}
% 	edge [pre] (c0)
% 	edge [post] (c1);

% 	\node [transition] at (1,-2) (t2) {$t_6$}
% 	edge [pre] (c1)
% 	edge [post] (c2);

% 	\node [transition] at (2,-1) (t5) {$t_8$}
% 	edge [pre] (c2)
% 	edge [post] (c3);
% \end{scope}

\draw[->] (4.5,-6) -- (6.5,-6) node [midway,sloped,swap,anchor=center,above] {iteration $t_7$};

\begin{scope}[shift={(8,-4)}]
	\node [place] at (-1,-1) (c0){};
	\node [above=0cm of c0] {$i$};
	\node [place] at (0,-2) (c1){};
	\node [below=0cm of c1] {$c_1$};
	\node [place] at (2,-2) (c2){};
	\node [below=0cm of c2] {$c_2$};
	\node [place] at (4,-2) (c3){};
	\node [above=0cm of c3] {$o$};

	\node [transition] at (-1,-2) (t1) {$t_1$}
	edge [pre] (c0)
	edge [post] (c1);

	\node [transition] at (1,-2) (t2) {$t_6$}
	edge [pre] (c1)
	edge [post] (c2);

	\node [transition] at (3,-2) (t5) {$t_5$}
	edge [pre] (c2)
	edge [post] (c3);
\end{scope}
\end{tikzpicture}
\label{fig:rules}
\caption{Example of rule applications}
\end{figure}

\section{Reduction Procedure}
\label{sec:procedure}

We show that the rules presented in the previous section 
summarize all sound FC-CWF nets in polynomial time. The proof is very 
involved, and we can only sketch it. 

We first show that acyclic FC-CWF nets can be completely reduced.

\begin{definition}[Graph]
The graph of a CWF net is the graph $(P\cup T, F)$. A CWF net is acyclic if its graph is acyclic.
\end{definition}

\begin{theorem}
\label{thm:acycliccomplete}
The merge and d-shortcut rule are complete for acyclic FC-CWF nets.
\end{theorem}

In the cyclic case we need the notion of {\em synchronizer of a loop}.
Although a similar concept was already used in \cite{negII}, the definition
there exploits the fact that negotiations are a structured model of communicating sequential agents. Since workflow nets do not have such a structure,
we need a different definition.

\begin{definition}[Loop]
Let $\W$ be a CWF net. A non-empty transition sequence 
$\sigma$ is a {\em loop} of $\W$ if $M \by{\sigma} M$ for some 
reachable marking $M$.
\end{definition}

\begin{definition}[Synchronizer]
Let $\W$ be a WF net. A free choice transition $t$ {\em synchronizes}
a loop $\sigma$ if $t$ appears in $\sigma$ and for every 
reachable marking $M$: if $M$ enables $t$, then $M(p) = 0$ for every
$p \in (\bigcup_{t' \in\sigma, t' \neq t}{}^\bullet t')$. A free choice 
transition is a {\em synchronizer} if it synchronizes some loop.
\end{definition}

Consider the insurance claim net, replacing the part between
the places $c_7 $ and $c_9$ by Figure \ref{fig:extendedClaim}. The sequence 
\texttt{process} \texttt{check1} \texttt{check2} \texttt{combine} \texttt{processing\_NOK} is a loop. Transitions \texttt{process}, \texttt{combine}, and \texttt{processing\_NOK} are synchronizers, but \texttt{check1} and \texttt{check2} are not. We use synchronizers to define fragments of $\W$ on which to apply our rules. 

\begin{definition}[Fragment]
Let $\W$ be a CWF net and let $t$ be a synchronizer of $\W$. 
The fragment $\W_t$ contains all transitions 
appearing in all loops synchronized by $t$, together with their input
and output places, and the arcs connecting them.
\end{definition}

In our example, the fragment $\W_{\rm process}$ is exactly the net of Figure \ref{fig:extendedClaim}. Our procedure selects a synchronizer $t$ and applies the rules to $\W_t$ until, loosely speaking, all loops synchronized by $t$ are removed from the net, and $t$ is no longer a synchronizer. The next lemma shows that when no synchronizers can be found anymore, the workflow net is acyclic, and so can be completely reduced by Theorem \ref{thm:acycliccomplete}.

\begin{lemma}
\label{lem:soundCyclicSync}
Every sound cyclic FC-CWF net has at least one synchronizer.
\end{lemma}
\begin{proofsketch}
We first show that in every sound cyclic FC-CWF net there exists a loop. We then inspect minimal loops and show that they must include a synchronizer. The proof constructs a transition sequence that pushes one token towards the final marking while all other tokens stay inside the loop. Should no synchronizer be present in the loop, this sequence ends in a dead lock contradicting soundness.
\end{proofsketch}

Given two synchronizers $t$ and $t'$, we say $\W_t \preceq \W_t'$ if every node of $\W_t$ is also a node of $\W_t'$. The relation $\preceq$ is a partial order on fragments. We have: 

\begin{lemma}
\label{lem:syncOnly}
Let $t$ be a synchronizer of a sound FC-CWF net. If $\W_t$ is minimal 
with respect to the partial order on fragments, then all non-synchronizers
of $\W_t$ can be removed by means of applications of the d-shortcut and merge rules.
\end{lemma}
\begin{proofsketch}
Intuitively, synchronizers are points where loops begin and end. For two distinct synchronizers of a {\em minimal} fragment, any occurrence sequence starting from the marking enabling one of them, ending in the marking enabling the other, and in which no other synchronizers occur, is acyclic. Thus we can reduce the possible paths from one synchronizer to another to a single transition using our rules. We do so by constructing auxiliary acyclic workflow nets and reducing those, applying the same reduction rules to our original net.
\end{proofsketch}

In our example, the fragment of Figure \ref{fig:extendedClaim} on the left is reduced
to the synchronizer-only fragment shown in Figure \ref{fig:extendedClaim} on the right. 
In such a fragment, a marking always marks exactly the places of one of 
the clusters, and nothing else. Intuitively, the
synchronizer-only fragment is an {\em S-net}, i.e., a net where 
every transition has exactly one input and one output place, but 
in which some places are {\em duplicated}. 
Figure \ref{fig:rules} shows an example of an S-net, while the net on the 
right of Figure \ref{fig:extendedClaim}  is an S-net in which place $c_{10}$ is 
duplicated in place $c_{11}$.

When reducing S-nets we must be careful that the shortcut rule does not ``run into cycles''. Consider for instance the second net in Figure \ref{fig:rules}. 
If instead of shortcutting $t_4$ we shortcut $t_1$, we obtain a new 
transition $t_7$ with $i$ and $c_2$ as input and output place.
If we now shortcut $t_7$, we return to the original net with an additional transition connecting $i$ and $o$. This problem is solved by imposing an (arbitrary) total order on the clusters. Using this order we classify transitions as ``forward'' (leading to a greater cluster) and ``backward'' (leading to a smaller cluster). Running into cycles is avoided by only applying the shortcut rule to the backward transition leading to a minimal cluster.
Ultimately, this procedure reduces the fragment to an acyclic net. The total number of synchronizers is thus reduced, until none are left. At this point, by Lemma \ref{lem:soundCyclicSync} the net is acyclic,
and Theorem \ref{thm:acycliccomplete} can be applied.  The complete reduction  algorithm is listed as Algorithm 1. The algorithm contains several points where the computation might end if some condition is fulfilled. If the net was free choice, we can then conclude that it is unsound.

\begin{algorithm}[t]
\caption{Reduction procedure for cyclic workflow nets $\W$}
\begin{algorithmic}[1]
\While{$\W$ is cyclic}
\State $c \gets$ a minimal synchronizer of $\W$  \Comment{If there is none, return}
\State $F \gets$ the fragment of $c$  \Comment{If fragment is malformed, return}
\While{$F$ contains non-synchronizers}
\State apply the merge rule exhaustively
\State apply the iteration rule exhaustively
\State apply the d-shortcut rule to $F$ \Comment{If not possible, return}
\EndWhile
\State fix a total order on $F$
\While{$F$ is cyclic}
\State apply the merge rule exhaustively
\State apply the iteration rule exhaustively
\State \algbox{apply the shortcut rule to the backward transition which ends at a minimal cluster}
\EndWhile
\EndWhile
\While{$\W$ is not reduced completely}
\State apply the merge rule exhaustively
\State apply the d-shortcut rule to $F$ \Comment{If neither was possible, return}
\EndWhile
\end{algorithmic}
\end{algorithm}

We have not yet discussed why a fragment could be malformed as mentioned in Line 3 of the algorithm. The proof that every minimal loop has a synchronizer also shows something more: tokens can only exit a loop at a cluster that contains a synchronizer, and all tokens exit the loop at the same time. Thus when we compute a fragment and find transitions that lead out of the fragment and whose cluster does not contain a synchronizer, or transitions that partially end outside and partially inside the fragment, we can already conclude that the net is unsound. For more information on how to compute fragments, see the next section.

With some analysis on the number of rule application in the acyclic case as well as the S-net case, we can bound the number of rule application to be polynomial:

\begin{theorem}
\label{thm:runtimeCyclic}
Every sound FC-CWF net can be summarized in at most $\mathcal{O}(|C|^4\cdot |T|)$ shortcut rule applications and $\mathcal{O}(|C|^4+|C|^2\cdot|T|)$ merge rule applications where $C$ is the set of clusters of the net.
Any unsound FC-CWF net can be recognized as unsound in the same time.
\end{theorem}

\subsection{Summarizing the example}
\label{subsec:summexample}

We illustrate our algorithm on the example of the insurance claim of Figure \ref{fig:claim}. To better illustrate our approach, we replace the part between
the places $c_7 $ and $c_9$ by Figure \ref{fig:extendedClaim}. 

Our algorithm begins by checking whether $\W$ is cyclic and finds a minimal synchronizer. This could in our example be $c_7$, its fragment is exactly the part of the net depicted in Figure \ref{fig:extendedClaim} on the left. Since the fragment contains non-synchronizers $c_{10},c_{11}$, the while loop of Line 4 is entered. The d-shortcut rule is applied to \texttt{check1} and \texttt{check2}. The resulting fragment is depicted in Figure \ref{fig:extendedClaim} on the right. This fragment consists only of synchronizers and thus the while loop ends. We fix as total order $[c_7] \prec [c_{10}]\prec [c_9]$.

Transition \texttt{processing\_NOK} is a backward transition as its post-set $[c_7]$ is smaller than its pre-set $[c_9]$ according to the total order. It is shortcut resulting in another backward transition ending in the cluster containing $c_{10},c_{11}$, which is then shortcut again to a self-loop on $c_9$. The self-loop is removed via the iteration rule.

The resulting net is depicted in Figure \ref{fig:extended2}. This net is acyclic, thus now the d-shortcut and merge rule are applied exhaustively. An intermediate step is depicted in Figure \ref{fig:extended3}. First \texttt{process\_questionnaire} and \texttt{time\_out} are merged and the path from $i$ to $c_5$ is shortcut. Then the linear path from $c_7$ to $o$ is shortcut into a single transition. Next the path from $i$ to $c_4$ is shortcut, resulting in the transition \texttt{register} to unconditionally enable \texttt{no\_processing} and \texttt{processing\_required}. Finally, with three more shortcuts and a merge, the net is completely reduced, and we obtain the transformer shown in Section \ref{subsec:summ}.

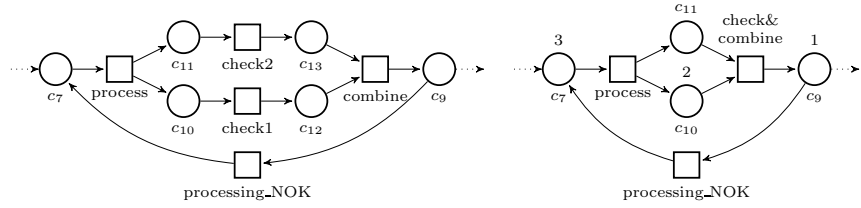
\begin{figure}[tb]
\centering
\scalebox{0.85}{
\begin{tikzpicture}[>=stealth',bend angle=45,auto,scale=1]
	\tikzstyle{every node}=[font=\scriptsize]
	\tikzstyle{place}=[circle,thick,draw=black,fill=white,minimum size=5mm,inner sep=0mm]
	\tikzstyle{transition}=[rectangle,thick,draw=black,fill=white,minimum size=4mm,inner sep=0mm]
	\tikzstyle{every label}=[black]
	
	\node [place] at (3,-1.5) (c7){};
	\node [below=0cm of c7] {$c_7$};
	\node [place] at (5,-2) (c10){};
	\node [below=0cm of c10] {$c_{10}$};
	\node [place] at (5,-1) (c11){};
	\node [below=0cm of c11] {$c_{11}$};
	\node [place] at (7,-2) (c12){};
	\node [below=0cm of c12] {$c_{12}$};
	\node [place] at (7,-1) (c13){};
	\node [below=0cm of c13] {$c_{13}$};
	\node [place] at (9,-1.5) (c9){};
	\node [below=0cm of c9] {$c_9$};

	\draw[dotted,->] (2.3,-1.5) -- (c7);
	\draw[dotted,->] (c9) -- (9.7,-1.5);

	\node [transition] at (4,-1.5) (t8) {}
	edge [pre] (c7)
	edge [post] (c10)
	edge [post] (c11);
	\node [below=0cm of t8] {process};

	\node [transition] at (6,-2) (t9) {}
	edge [pre] (c10)
	edge [post] (c12);
	\node [below=0cm of t9] {check1};

	\node [transition] at (6,-1) (t10) {}
	edge [pre] (c11)
	edge [post] (c13);
	\node [below=0cm of t10] {check2};

	\node [transition] at (8,-1.5) (t11) {}
	edge [pre] (c12)
	edge [pre] (c13)
	edge [post] (c9);
	\node [below=0cm of t11] {combine};

	\node [transition] at (6,-3) (t12) {}
	edge [pre,bend right=20] (c9)
	edge [post,bend left=20] (c7);
	\node [below=0cm of t12] {processing\_NOK};

\end{tikzpicture}
}
\scalebox{0.85}{
\begin{tikzpicture}[>=stealth',bend angle=45,auto,scale=1]
	\tikzstyle{every node}=[font=\scriptsize]
	\tikzstyle{place}=[circle,thick,draw=black,fill=white,minimum size=5mm,inner sep=0mm]
	\tikzstyle{transition}=[rectangle,thick,draw=black,fill=white,minimum size=4mm,inner sep=0mm]
	\tikzstyle{every label}=[black]
	
	\node [place] at (3,-1.5) (c7){};
	\node [below=0cm of c7] {$c_7$};
	\node [above =0cm and 0cm of c7] {3};
	\node [place] at (5,-2) (c10){};
	\node [below=0cm of c10] {$c_{10}$};
	\node [above =0cm and 0cm of c10] {2};
	\node [place] at (5,-1) (c11){};
	\node [above=0cm of c11] {$c_{11}$};
	% \node [above =0cm and 0cm of c11] {2};
	\node [place] at (7,-1.5) (c9){};
	\node [below=0cm of c9] {$c_9$};
	\node [above =0cm and 0cm of c9] {1};

	\draw[dotted,->] (2.3,-1.5) -- (c7);
	\draw[dotted,->] (c9) -- (7.7,-1.5);

	\node [transition] at (4,-1.5) (t8) {}
	edge [pre] (c7)
	edge [post] (c10)
	edge [post] (c11);
	\node [below=0cm of t8] {process};

	\node [transition] at (6,-1.5) (t11) {}
	edge [pre] (c10)
	edge [pre] (c11)
	edge [post] (c9);
	\node [above=0cm of t11] {\begin{tabular}{c}
		check\&\\combine
	\end{tabular}};

	\node [transition] at (5,-3) (t12) {}
	edge [pre,bend right=20] (c9)
	edge [post,bend left=20] (c7);
	\node [below=0cm of t12] {processing\_NOK};

\end{tikzpicture}
}
\caption{Extension of the insurance claim net and the synchronizer-only fragment}
\label{fig:extendedClaim}
\end{figure}

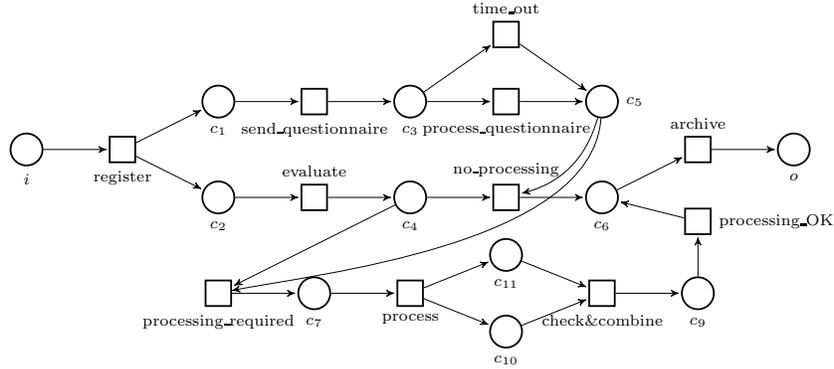
\begin{figure}[tb]
\centering
\scalebox{0.85}{
\begin{tikzpicture}[>=stealth',bend angle=45,auto,scale=1.5]
	\tikzstyle{every node}=[font=\scriptsize]
	\tikzstyle{place}=[circle,thick,draw=black,fill=white,minimum size=5mm,inner sep=0mm]
	\tikzstyle{transition}=[rectangle,thick,draw=black,fill=white,minimum size=4mm,inner sep=0mm]
	\tikzstyle{every label}=[black]
	
	\node [place] at (0,0) (i){};
	\node [below=0cm of i] {$i$};
	\node [place] at (2,0.5) (c1){};
	\node [below=0cm of c1] {$c_1$};
	\node [place] at (2,-0.5) (c2){};
	\node [below=0cm of c2] {$c_2$};
	\node [place] at (4,0.5) (c3){};
	\node [below=0cm of c3] {$c_3$}; 
	\node [place] at (4,-0.5) (c4){};
	\node [below=0cm of c4] {$c_4$};
	\node [place] at (6,0.5) (c5){};
	\node [right=0cm of c5] {$c_5$};
	\node [place] at (6,-0.5) (c6){};
	\node [below=0cm of c6] {$c_6$};
	\node [place] at (3,-1.5) (c7){};
	\node [below=0cm of c7] {$c_7$};
	\node [place] at (5,-1.9) (c10){};
	\node [below=0cm of c10] {$c_{10}$};
	\node [place] at (5,-1.1) (c11){};
	\node [below=0cm of c11] {$c_{11}$};
	\node [place] at (7,-1.5) (c9){};
	\node [below=0cm of c9] {$c_9$};
	\node [place] at (8,0) (o){};
	\node [below=0cm of o] {$o$};

	\node [transition] at (1,0) (t1) {}
	edge [pre] (i)
	edge [post] (c1)
	edge [post] (c2);
	\node [below=0cm of t1] {register};

	\node [transition] at (3,0.5) (t2) {}
	edge [pre] (c1)
	edge [post] (c3);
	\node [below=0cm of t2] {send\_questionnaire};

	\node [transition] at (5,0.5) (t3) {}
	edge [pre] (c3)
	edge [post] (c5);
	\node [below=0cm of t3] {process\_questionnaire};

	\node [transition] at (5,1.2) (t4) {}
	edge [pre] (c3)
	edge [post] (c5);
	\node [above=0cm of t4] {time\_out};

	\node [transition] at (3,-0.5) (t5) {}
	edge [pre] (c2)
	edge [post] (c4);
	\node [above=0cm of t5] {evaluate};

	\node [transition] at (5,-0.5) (t6) {}
	edge [pre] (c4)
	edge [pre,bend right=30] (c5)
	edge [post] (c6);
	\node [above=0cm of t6] {no\_processing};

	\node [transition] at (2,-1.5) (t7) {}
	edge [pre] (c4)
	edge [post] (c7);
	\node [below=0cm of t7] {processing\_required};

	\draw[->] (c5) .. controls (5.9, -1) and (3, -1.3).. (t7);

	\node [transition] at (4,-1.5) (t8) {}
	edge [pre] (c7)
	edge [post] (c10)
	edge [post] (c11);
	\node [below=0cm of t8] {process};

	\node [transition] at (6,-1.5) (t11) {}
	edge [pre] (c10)
	edge [pre] (c11)
	edge [post] (c9);
	\node [below=0cm of t11] {check\&combine};

	\node [transition] at (7,-0.75) (t11) {}
	edge [pre] (c9)
	edge [post] (c6);
	\node [right=0cm of t11] {processing\_OK};

	\node [transition] at (7,0) (t12) {}
	edge [pre] (c6)
	edge [post] (o);
	\node [above=0cm of t12] {archive};

\end{tikzpicture}
}
\caption{After shortcutting backward transitions}
\label{fig:extended2}
\end{figure}

\begin{figure}[tb]
\centering
\scalebox{0.85}{
\begin{tikzpicture}[>=stealth',bend angle=45,auto,scale=1.5]
	\tikzstyle{every node}=[font=\scriptsize]
	\tikzstyle{place}=[circle,thick,draw=black,fill=white,minimum size=5mm,inner sep=0mm]
	\tikzstyle{transition}=[rectangle,thick,draw=black,fill=white,minimum size=4mm,inner sep=0mm]
	\tikzstyle{every label}=[black]
	
	\node [place] at (0,0) (i){};
	\node [below=0cm of i] {$i$};
	\node [place] at (2,-0.5) (c2){};
	\node [below=0cm of c2] {$c_2$};
	\node [place] at (4,-0.5) (c4){};
	\node [below=0cm of c4] {$c_4$};
	\node [place] at (4,0.5) (c5){};
	\node [below=0cm of c5] {$c_5$};
	\node [place] at (6,0.5) (c6){};
	\node [below=0cm of c6] {$c_6$};
	\node [place] at (6,-0.5) (c7){};
	\node [below=0cm of c7] {$c_7$};
	\node [place] at (8,0) (o){};
	\node [below=0cm of o] {$o$};

	\node [transition] at (1,0) (t1) {}
	edge [pre] (i)
	edge [post] (c5)
	edge [post] (c2);
	\node [below=0cm of t1] {\begin{tabular}{c}register\&\\questionnaire\end{tabular}};

	\node [transition] at (3,-0.5) (t5) {}
	edge [pre] (c2)
	edge [post] (c4);
	\node [above=0cm of t5] {evaluate};

	\node [transition] at (5,0.5) (t6) {}
	edge [pre] (c4)
	edge [pre] (c5)
	edge [post] (c6);
	\node [above=0cm of t6] {no\_processing};

	\node [transition] at (5,-0.5) (t7) {}
	edge [pre] (c4)
	edge [pre] (c5)
	edge [post] (c7);
	\node [below=0cm of t7] {processing\_required};

	\node [transition] at (7,-0.5) (t8) {}
	edge [pre] (c7)
	edge [post] (o);
	\node [below=0cm of t8] {process\&archive};

	\node [transition] at (7,0.5) (t12) {}
	edge [pre] (c6)
	edge [post] (o);
	\node [above=0cm of t12] {archive};

\end{tikzpicture}
}
\caption{After some rule applications}
\label{fig:extended3}
\end{figure}
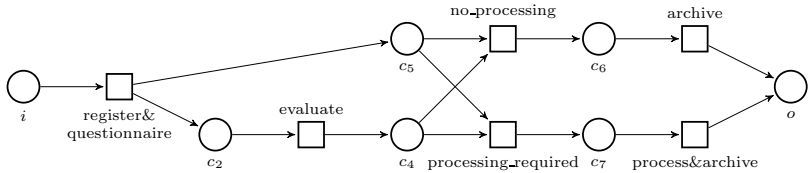

\subsection{Extension to generalized soundness}

In \cite{DBLP:conf/apn/HeeSV04,DBLP:conf/apn/2004} (see also \cite{van2011soundness}), 
two alternative notions of soundness are introduced: $k$-soundness and generalized soundness. 
We show that for free choice workflow nets they coincide with the standard notion. 
Therefore, our rules are also complete with respect to these alternative notions. 

\begin{definition}
Let $\W=(P,T,F,i,o)$ be a workflow net. For every $k \geq 1$, let
$i^k$ ($o^k$) denote the marking that puts $k$ tokens on $i$ (on $o$), and no tokens elsewhere. $\W$ is {\em $k$-sound} if $o^k$ is reachable from every marking reachable from $i^k$. $\W$ is {\em generalized sound} if it is $k$-sound for every $k \geq 1$.
\end{definition}

\begin{theorem}
\label{thm:kSound}
Let $\W$ be a free choice workflow net. The following statements are equivalent: (1) $W$ is sound; (2) $W$ is $k$-sound for some $k \geq 1$; (3) $W$ is generalized sound.
\end{theorem}

% Interestingly, we do not yet know if our rules preserve both $k$-soundness {\em and} the summary. Observe that the notion of summary has to be generalized from a relation
% $S \subseteq {\cal M}_i \times {\cal M}_o$ to a relation $S \subseteq {\cal M}_i^k  \times {\cal M}_o^k$. We leave this question for further research.

\section{Experimental evaluation}
\label{sec:experiments}

\setlength\textfloatsep{10pt}

We have implemented our reduction algorithm and applied it to a benchmark suite of 
models previously studied in \cite{van2007verification,fahland2009instantaneous}.

\footnote{Nets can be obtained under \url{http://svn.gna.org/viewcvs/*checkout*/service-} \url{tech/trunk/_meta/nets/challenge/} in folders \texttt{sap-reference} and \texttt{ibm-soundness}}
The most complex part of the implementation\footnote{Can be obtained under \url{https://www7.in.tum.de/tools/workflow/index.php}} is the computation of synchronizers and their fragments. A crucial point is that we are only interested in fragments that consist of free choice places as those are the fragments we might be able to completely reduce.
The computation of the synchronizers starts with an overapproximation: starting from a cluster $c$, we begin by marking for all transitions $t \in c_T$, the places in $t^\bullet$ that are free choice as visited. Whenever we have marked all places in a cluster as visited, we repeat the same for this cluster. In that way we overapproximate the set of clusters that can occur in an occurrence sequence as in the definition of synchronizer. Should all places in $c$ be marked as visited at some point, we consider $c$ a potential synchronizer.

We now compute the fragment of $c$ in a backwards fashion. Starting with only $c$, we check for every transition whose out-places are contained
 in the currently identified fragment, whether its in-places 
were completely marked in the first step. If so, add its in-places and the transition to the fragment.
We also check simple soundness properties, e.g. that no transition exists which 
starts in the fragment and ends partially inside and partially outside the fragment.

\begin{table}[tb]
\def\arraystretch{1.4}%
\setlength\tabcolsep{5pt}
\setlength\abovecaptionskip{10pt}
\centering
\scalebox{0.9}{
\begin{tabular}{l|r|r|r|r|r|r|r|r|r}
 & \multicolumn{1}{c|}{\#} & \multicolumn{3}{c|}{$|P|$} & \multicolumn{3}{c|}{$|T|$} & red. & \# rule \\
 & nets & avg.&med. & max & avg.& med. & max & by & appl.\\
\hline
Acyclic FC sound & 446 & 20.7 & 13 & 154 & 13.1 & 9 & 95 & --- & 12.8 \\
Acyclic FC uns. & 761 & 60.4 & 49 & 264 & 41.1 & 33 & 285 & 73.6\% & 38.0\\
Cyclic FC sound& 24 & 46.1 & 43 & 118 & 34.3 & 26 & 93 & --- & 43.2 \\
Cyclic FC uns. & 155 & 73.2 & 61 & 274 & 51.1 & 44 & 243 & 78.1\% & 53.2 \\
Acyclic not FC & 542 & 47.0 & 38 & 262 & 46.8 & 37 & 267 & 68.4\% & 38.4\\
Cyclic not FC & 30 & 85.6 & 72 & 193 & 88.1 & 72 & 185 & 66.4\% & 82.7 \\
\end{tabular}
}%end of scalebox
\caption{Analyzed workflow nets}
\label{tab:experiments}
\end{table}

%\begin{table}
%\def\arraystretch{1.5}%
%\setlength\tabcolsep{7pt}
%\centering
%\begin{tabular}{ccc}
% & FC & not FC\\
%Acyclic & 45.7 & 47\\
%Cyclic & 69.5 & 85.6\\
%\end{tabular}
%\hspace{1cm}
%\begin{tabular}{ccc}
% & FC & not FC\\
%Acyclic & 9.4 (20.70\%) & 15.6 (33.30\%)\\
%Cyclic & 20.6 (29.60\%) & 31.8 (37.20\%)\\
%
%\end{tabular}
%\caption{Avg. places per net before and after reduction (percentage)}
%\label{tab:experiments2}
%\end{table}

We have conducted some experiments to obtain answers to the following two questions:
 (1) Since our rules must preserve not only soundness, 
but also the input/output relation, they cannot be as ``aggressive'' as previous ones. So 
it could be the case that they only lead to a small reduction factor
in the non-free choice case. To explore this question, we experimentally compute the reduction 
factor for non-free choice benchmarks.
 (2) While Theorem \ref{thm:runtimeCyclic} is a strong theoretical result (compared to PSPACE-hardness of soundness for arbitrary workflow nets), the $\mathcal{O}(|C|^4\cdot |T|)$ bound 
has rather high exponents, and could potentially lead to an impractical reduction algorithm. To
explore if the worst case appears in practice, we compute the number of rule 
applications for free choice benchmarks.

We have used the benchmark suites of \cite{van2007verification,fahland2009instantaneous}, both consisting of industrial examples.
We analyzed a total of 1958 nets, of which 1386 were free choice. 
Running the reduction procedure for all benchmarks took 6 seconds. 
The results are shown in Table \ref{tab:experiments}.
The number of places and 
transitions are always given as average/median/max. In the free choice case, our algorithm found 
that 470 nets were sound (i.e. those nets were reduced completely), and on average the nets were 
reduced to about 23\% of their original size. In the non-free choice case no net could be 
reduced completely (which does not necessarily mean they are all unsound). However, the size 
of the nets was still reduced to about 35\% of their original size. While we have omitted some more data on the number of rule applications due to lack of space, our experiments indicate that the number of rule applications is close to linear in the size of the net. 

\section{Conclusion}
\label{sec:conclusion}

We have presented the first set of reduction rules for 
colored workflow nets that preserves not only soundness, but also the 
input/output relation, and is complete for free choice nets. We have also designed a specific reduction algorithm. Experimental results for 1958 workflow nets derived 
from industrial business processes show that the nets are reduced to 
about 30\% of their original size.  

Our rules can be used to prove properties of the input/output relation
by computing it. To reduce the complexity of the computation, we observe that our reduction rules are easily compatible with abstract 
interpretation techniques: given an abstract domain of data values, the rules 
can be adapted so that, instead of computing the transformers of the 
new transitions using the union, join, and Kleene-star operators, they 
compute their abstract versions. We plan to study this combination in future research.

\paragraph{Acknowledgements} Thank you very much to Karsten Wolf for pointing us to the benchmarks. Many thanks to the anonymous reviewers for the helpful comments.

%\bibliography{ref}

\setlength{\textfloatsep}{\saveold}

\appendix
\section{Appendix}
\subsection{Proof of Section \ref{sec:rules}}

\begin{reftheorem}{thm:correctness}
The merge, shortcut and iteration rules are correct for CWF nets.
\end{reftheorem}

\begin{proof}
Correctness of the merge and iteration rules follows easily from the definitions.
We concentrate on the shortcut rule, whose proof is more delicate.  

Assume the transition $t$ and cluster $c$ are as in Definition \ref{def:shortcut}. 
We say that $c$ occurs in a firing sequence $\sigma$ if some transition $t\in c\cap T$ occurs in $\sigma$.

Let $\W_2$ be the result of applying the shortcut rule to $\W_1$. The proof is divided into four parts. 

\begin{itemize}
\item[(1)] For every initial firing sequence $\sigma_2$ of $\W_2$, there is an initial  firing
sequence $\sigma_1$ of $\W_1$ leading to the same marking. 

Let $\sigma_1$ be the result of replacing all occurrences of $t'_s$ (as in Definition \ref{def:shortcut}) in $\sigma_2$ by the sequence $t \, t'$. Clearly this yields an initial firing sequence $\sigma_1$ of $\W_1$. 
The marking reached by these two sequences is the same.

\item[(2)] For each initial firing sequence $\sigma_1$ of $\W_1$, there is an initial  firing
sequence $\sigma_2$ of $\W_1$ leading to the same marking. 

Since $t$ unconditionally enables $c$, between any two occurrences of $t$ in $\sigma_1$ 
there is an occurrence of some $t' \in c\cap T$. Define $\sigma_2$ as the result of, for each occurrence
of $t$, removing the next occurrence of $t'\in c\cap T$, and replacing $t$ by an occurrence of $t'_s$.
Should the last occurrence of $t$ not be followed by an occurrence of $c$, 
we instead start with $\sigma_1t'$ for some $t'$. 
This yields a firing sequence $\sigma_2$ in $\W_2$, 
the corresponding sequence to $\sigma_1$, 
which leads to the same marking as $\sigma_1$ (or $\sigma_1t'$).

\item[(3)] If $\W_1$ is sound then $\W_2$ is sound. 

We first show that every transition in $\W_2$ can be enabled by some initial firing sequence.
Since every transition in $\W_1$ can be enabled by some initial firing sequence, using the corresponding sequence in $\W_2$, which exists by (2), we are done for all transitions but those in $c\cap T$.
If $c$ still exists in $\W_2$, there must be some $t''\in T$ such that such that $t''\neq t$ and $p \in t''^\bullet$ for some $p \in c\cap P$. Since $t''\neq t$, the transition $t''$ is unchanged in $\W_2$. By soundness of $W_1$, 
some initial firing sequence $\sigma_1$ enables $t''$. We extend it by an occurrence of $t''$ and then to a firing sequence $\sigma_1t''\rho_1$ leading to the final marking in $\W_1$, which is possible since $\W_1$ is sound. This sequence  contains an occurrence of $c$ which is not matched by a prior occurrence of $t$, and so does the corresponding sequence in $\W_2$. Together with the fact that $c$ is a free choice cluster, it follows that all transitions in $c\cap T$ can be enabled in $\W_2$.

We now prove that every initial firing sequence $\sigma_2$ in $\W_2$ can be extended to a sequence that ends with the final marking. Take the corresponding occurrence sequence $\sigma_1$ in $\W_1$, and extend it to a sequence $\tau_1 = \sigma_1\rho_1$ that ends with the final marking in $\W_1$ (possible by soundness of $\W_1$). The corresponding sequence in $\W_2$ is $\tau_2=\sigma_2\rho_2$, which is the extension of $\sigma_2$ (by construction of corresponding sequences) that ends with the final marking.

\item[(4)] If $\W_2$ is sound then $\W_1$ is sound. 

Since every transition in $\W_2$ can be enabled by some initial occurrence sequence, using the corresponding sequence in $\W_1$ we see that the same is true for all transitions but those in $c\cap T$. However, it is then easy to show that those transitions can be enabled in $\W_1$: take the initial firing sequence that enables $t$ and extend it by $t$ which unconditionally enables $c$.

For an initial occurrence sequence $\sigma_1$ in $\W_1$, the corresponding occurrence sequence $\sigma_2$ in $\W_2$ can be extended to a sequence $\tau_2 = \sigma_2\rho_2$ that ends with the final marking in $\W_2$. The corresponding sequence $\tau_1$ in $\W_1$ is either $\sigma_1\rho_1$ or $\sigma_1t'\rho_1$ for some $t' \in c\cap T$, an extension of $\sigma_1$ that ends with the final marking.

\end{itemize}
\end{proof}

%%%%% HERE IS THE ACTUAL PROOF %%%%%
\subsection{Proof of Section \ref{sec:procedure}}

We recall some important results from \cite{desel2005free} about free choice nets and S-components. 
Recall that a net is like a workflow net, but without a distinguished input and output place (see e.g. \cite{desel2005free}). A Petri net is a pair $(N, M_0)$, where $N$ is a net and $M_0$ is a  marking of $N$ called the initial marking.

\begin{definition}[S-component]
Let $N=(P,T,F)$ be a net. An S-component of $N$ is a net $N'=(P',T',F')$ such that
\begin{itemize}
\item $\emptyset \neq P' \subseteq P$
\item ${}^\bullet s \cup s^\bullet \subseteq T'$ for every $s \in P'$
\item $N'$ is a strongly connected S-net (every transition has exactly one pre- and post place in $N'$)
\end{itemize}
\end{definition}

We recall the following theorem: 

\begin{theorem}[S-component coverability]\cite{desel2005free}
A live and bounded free choice Petri net $N$ can be covered by S-components, i.e. there are S-components 
of $N$ such that every place belongs to some S-component.
\end{theorem}

This has implications for sound free choice workflow nets $\W$: 

\begin{definition}
For a workflow net $\W$, let $\oline{\W}$ be the result of adding to 
$\W$ a transition $t^*$ with $o$ as input and $i$ as output place. 
We call $\oline{\W}$ the extended net of $\W$.
\end{definition}

Since $W$ is sound, the Petri net $(\oline{\W}, i)$ (i.e., the net $\oline{\W}$ with
initial marking $i$) is live and bounded \cite{DBLP:journals/jcsc/Aalst98} 
and thus can be covered by S-components.

\begin{reftheorem}{thm:acycliccomplete}
The merge and d-shortcut rule are complete for acyclic FC-CWF nets.
\end{reftheorem}

We split the proof into three lemmas from which the result follows. 
We call a net {\em irreducible} if none of our rules is applicable. 

\begin{lemma}
\label{lemma1}
Let $\W$ be an sound FC-CWF net that is irreducible and 
let $p \in P$ be a place of $\W$ with $|p^\bullet| > 1$.
Then every S-component of $\oline{\W}$ contains an element in $[p]\cap P$. 
\end{lemma}

\begin{proof}
We proceed in two steps.\\[0.2cm]
\noindent (a) There is a transition t in $p^\bullet$ such that: either $t^\bullet = o$ or $[q]\cap P \subseteq t^\bullet$ for some $q \in P$ with $|q^\bullet| > 1$.

This is the core of the proof. We first claim: if $[q]\cap P \subseteq t^\bullet$ for some $t \in [p]\cap T, q \in P$, then (a) holds. 
Indeed: if $[q]\cap P \subseteq t^\bullet$ for some $t \in [p]\cap T, q \in P$, then either $q = o$ or $|q^\bullet| > 1$, because otherwise the d-shortcut rule can be applied to 
$p, t$ and $q$, contradicting the irreducibility of $\W$. This proves the claim. 

It remains to prove that $[q]\cap P \subseteq t^\bullet$ for some $t \in [p]\cap T, q \in P$. For this,
we assume the contrary, and prove that $\W$ contains a cycle, contradicting the 
hypothesis.

Since the merge rule is not applicable to $\W$, 
$p^\bullet$ contains two transitions $t_1, t_2$ such that $t_1^\bullet \neq t_2^\bullet$. We proceed in three steps.

\begin{itemize}
\item[(a1)] For every reachable marking $M$ that marks $[p]\cap P$ there is a
sequence $\sigma$ such that $M \by{t_1 \sigma} M_1$ and $M \by{t_2 \sigma} M_2$ for some markings $M_1, M_2$, and the sets $A_1$ and $A_2$ of clusters marked by $M_1, M_2$ are disjoint.

Let $\sigma$ be a longest occurrence sequence such that $M \by{t_1 \sigma} M_1$ 
and $M \by{t_2 \sigma} M_2$ for some markings $M_1, M_2$ (notice that
$\sigma$ exists, because all occurrence sequences of $\W$ are finite by acyclicity).
We have $A_1 \cap A_2 = \emptyset$, because otherwise we can extend $\sigma$ by firing any transition of any cluster marked by both markings.
%We prove that, furthermore, $A_1 \neq \{o\} \neq A_2$. Assume w.l.o.g. $A_1 = \{o\}$.
%Then, since $\W$ is sound, we have $M_1 = M_f$, which means that the last 
% step of $\sigma$ is of the form $(n_f, r_f)$. So $M_2$ is also a marking obtained after 
% the occurrence of $(n_f,r_f)$. Since every agent participates in $n_f$ and $\trans(n_f,a,r_f)=\emptyset$ for every
% agent $a$ and transition $r_f$, we also have $M_2 = M_f$. So
% $M \by{(n,r_1)} M_1' \by{\sigma} M_f$ and $M \by{(n,r_1)} M_2' \by{\sigma} M_f$, which implies $M_1' = M_2'$. Since $\N$ is deterministic, we then have $\trans(n,a,r_1) = \trans(n,a,r_2)$, contradicting the hypothesis.

\item[(a2)] For every $a_1 \in A_1$ there is a path leading from some $a_2 \in A_2$ to $a_1$, and
for every $a_2 \in A_2$ there is a path leading from some $a_1 \in A_1$ to $a_2$.

By symmetry it suffices to prove the first part. Since $A_1$ and $A_2$ are disjoint, 
$a_1$ is marked by $M_1$ but not by $M_2$. Thus there is a place $q$ in $a_1$ that is not marked by $M_2$. 
%not true if multiple tokens
%Moreover, since $\W$ is acyclic, every cluster can fire at most once in an occurrence sequence, and so neither a transition of $a_1$ nor one from $a_2$ appear in $\sigma$.
Since, the sequences $t_1 \sigma$ and $t_2 \sigma$ only differ in their first element, 
it must hold that $q\in t_1^\bullet$ and $q\not\in t_2^\bullet$. Let $S$ be an S-component of $\oline{\W}$ that contains $q$. Then $S$ contains $^\bullet q$ and thus $t_1$. It also contains some place in $a_{1P} = {}^\bullet t_1$ and therefore also each transition in $a_{1T}$ and in particular $t_2$. In $S$, let $q'$ be the target of $t_2$. It is clear that $q' \neq q$ and even $q' \notin a_{1p}$ as an S-component can only contain one place per cluster.

Comparing the markings after $t_1$ fired and after $t_2$ fired, the token of $S$ is in $q$ in the first case and in $q'$ in the second case. This token will in both cases remain there during the sequence $\sigma$, thus $M_2$ marks $q'$ while $M_1$ marks $q$.
%TODO: Adjust Figure
(see Figure \ref{fig:polyproof}).

\begin{figure}[h]
%\centerline{\scalebox{0.45}{\input{polyproof.pdf_t}}}
\centering
\begin{tikzpicture}[decoration=snake]
%places of \oline{p}
\node[circle,draw] (p_P1) at (0,0) {};
\node[circle,draw] (p_P2) at (1,0) {};
\node[circle,draw] (p_P3) at (2,0) {};
\node[left = -0.1cm of p_P3] {$p$};
%transitions t_1,t_2
\node[rectangle,draw] (p_T1) at (0.5,-1) {};
\node[left = 0cm of p_T1] {$t_1$};
\node[rectangle,draw] (p_T2) at (1.5,-1) {};
\node[right = 0cm of p_T2] {$t_2$};
%places of a_2
\node[circle,draw] (a2_P1) at (2.5,-2) {};
\node[circle,draw] (a2_P2) at (3.5,-2) {};
\node[circle,minimum size=1.2mm,inner sep=0,draw,fill=white] at (a2_P1.center) {};
%places of a_1
\node[circle,draw] (a1_P1) at (-0.5,-4) {};
\node[circle,draw] (a1_P2) at (0.5,-4) {};
\node[below left = -0.15cm and 0cm of a1_P2] {$q$};
%tokens
\node[circle,minimum size=1.2mm,inner sep=0,draw,fill=white] at ($(a1_P1.center)+(0.06,0.06)$) {};
\node[circle,minimum size=1.2mm,inner sep=0,draw,fill=black] at ($(a1_P1.center)+(-0.06,-0.06)$) {};
\node[circle,minimum size=1.2mm,inner sep=0,draw,fill=black] at (a1_P2.center) {};
%places of \oline{q'}
\node[circle,draw] (q_P1) at (1.5,-4) {};
\node[circle,draw] (q_P2) at (2.5,-4) {};
\node[below right = -0.3cm and 0cm of q_P1] {$q'$};
\node[circle,minimum size=1.2mm,inner sep=0,draw,fill=white] at (q_P1.center) {};
%boxes
\node[rectangle,draw,fit={(p_P1) (p_P3) (p_T1)}] {};
\node[rectangle,draw,fit={($(a1_P1.north west)+(0pt,5pt)$) ($(a1_P2.south east)+(0pt,-15pt)$)}] (boxa1){};
\node[left=0cm of boxa1]{$a_1$};
\node[rectangle,draw,fit={($(a2_P1.north west)+(0pt,5pt)$) ($(a2_P2.south east)+(0pt,-15pt)$)}](boxa2){};
\node[right=0cm of boxa2]{$a_2$};
\node[rectangle,draw,fit={($(q_P1.north west)+(0pt,5pt)$) ($(q_P2.south east)+(0pt,-15pt)$)}] {};
%explanation
\node[circle,minimum size=1.2mm,inner sep=0,draw,fill=black] (tokenBlack) at (5,-3) {};
\node[circle,minimum size=1.2mm,inner sep=0,draw,fill=white] (tokenWhite) at (5,-3.5) {};
\node[right=0cm of tokenBlack] {Marking $M_1$};
\node[right=0cm of tokenWhite] {Marking $M_2$};
\pgfsetarrowsend{latex}
\draw (p_P1) -- (p_T1);
\draw (p_P2) -- (p_T1);
\draw (p_P3) -- (p_T1);
\draw (p_P1) -- (p_T2);
\draw (p_P2) -- (p_T2);
\draw (p_P3) -- (p_T2);
\draw (-1,-3) -- (a1_P1);
\draw (p_T1) -- (a1_P2);
\draw (p_T2) -- (q_P1);
\draw (p_T2) -- (a2_P1);
\draw (4,-1) -- (a2_P2);
%squiggly
\draw[decorate] (a2_P1) -- (q_P2);
%\draw[decorate] (q_P1.south) arc (180:0:-0.5);
\draw[decorate,bend left=73] (q_P1) to (a1_P2);
\end{tikzpicture}
\caption{Illustration of the proof of Lemma \ref{lemma1}.}
\label{fig:polyproof}
\end{figure}
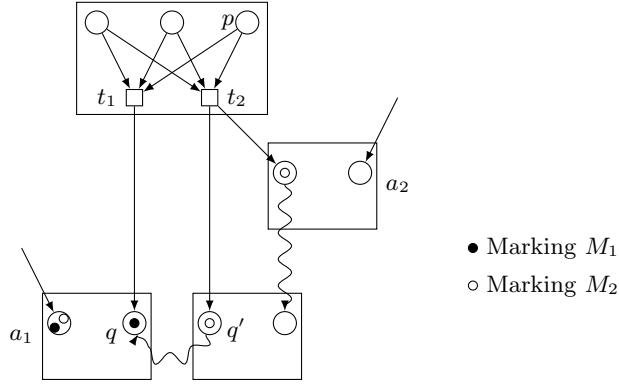

We first show that there is a path from $[q']$ to $[q]$.
By assumption, there is no transition of $[p]\cap T$ such that $[q]\cap P \subseteq t^\bullet$ for some $q \in P$, and so 
$[q]\cap P \not\subseteq t_1^\bullet$.
Thus $[q]\cap P$ contains a place $r \neq q$ such that $r\not\in t_1^\bullet$, and since $M_1$ marks $[q]$ (and therefore $r$), it holds that either some transition in $\sigma$ marked $r$ or $r$ was already marked by $M$. Therefore $M_2$ must also mark $r$.

Since $M_2$ marks $r$, and $\W$ is sound, there is a sequence of transitions
$\tau$ such that $M_2 \by{\tau} M_2'$ and $M_2'$ marks $[r]$. Since in $M_2$, the token in the S-component of $q$ is on $q'$, there is a path from $q'$ to $q$ and thus from $[q']$ to $[q]$.

We now prove that there is a path from some $a_2 \in A_2$ to $[q']$. If $[q']$ is marked by $M_2$, then $[q'] \in A_2$ and we are done. If $[q']$ is not enabled at $M_2$ (as in the figure) then, since $M_2$ marks $q'$ and $\W$ is sound,
there is a sequence of transitions $\tau$ such 
that $M_2 \by{\tau} M_2'$ and $M_2'$ enables $[q']$. 
Since $A_2$ is the set of clusters enabled at $M_2$, we have $\tau = t \, \tau'$ for some transition $t$ with $[t]\in A_2$. So some subword of $\tau$ is a path from some cluster of $a_2$ to $[q']$.

\item[(a3)] ${\cal N}$ contains a cycle. 

Follows immediately from (a2) and the finiteness of $N_1$ and $N_2$.
\end{itemize}

\noindent (b) Every every S-component of $\oline{\W}$ contains an element in $[p]\cap P$. \\
By repeated application of (a) we find a chain of clusters and transitions $(a_1, t_1) \ldots (a_k, t_k)$ such that
$a_1 = [p]$, $a_k = o$, $t_i\in{a_i}\cap T$ and ${a_{i+1}}\cap T \subseteq t_i^\bullet)$ for every $1 \leq i \leq k-1$. Since every S-component contains $o$, every S-component must contain $t_{n-1}$ and also an element of $a_{n-1}$, and also an element of $a_{n-2},\ldots,a_1$.
\end{proof}

\begin{lemma} 
\label{lemma2} 
Let $\W$ be a sound acyclic FC WF net that is irreducible. Every S-component contains a place of every cluster,
and for every cluster $a \neq [o]$ and every transition $t$ of $a\cap T$ there is a cluster $a'$ satisfying 
$t^\bullet =a'\cap P$. 
\end{lemma} 

\begin{proof}
We first show that every S-component contains a place of every cluster. By Lemma \ref{lemma1},
it suffices to prove that every cluster $a \neq [o]$ contains more than one transition.
Assume the contrary, i.e., some cluster different from $[o]$ contains only one transition.
Since, by soundness, every transition can occur, there is an occurrence sequence
$(a_0, t_0) (a_1, t_1) \cdots (a_k, t_k)$ such that $a_k$ contains only one transition
and all of $a_0, \ldots, a_{k-1}$ contain more than one transition. By Lemma \ref{lemma1}, 
every S-component contains a place in all of $a_0, a_1, a_{k-1}$. It follows that ${a_{i+1}}\cap T \subseteq t_i^\bullet$ 
$0 \leq i \leq k-1$. 
In particular, ${a_{k}}\cap T \subseteq t_{k-1}^\bullet$. But 
then, since $a_k$ only has one transition, the d-shortcut rule can be applied to 
$a_{k-1}, t_{k-1}, a_k$, contradicting the hypothesis that $\W$ is irreducible.

For the second part, assume there is a cluster $a \neq [o]$, a transition $t$ of $a$, and two clusters $a_1,a_2$ such that ${a_1}\cap P \cap t^\bullet \neq \emptyset \neq {a_2}\cap P \cap t^\bullet$. Let $p_1$ be some place in ${a_1}\cap P \cap t^\bullet$ and $p_2$ some place in ${a_2}\cap P \cap t^\bullet$
By the first part,
every S-component contains a place in $a$, $a_1$ and $a_2$. Since $\W$ is sound, some reachable
marking $M$ marks $a$. Moreover, since all S-components contain a place in $a$, and every S-component contains exactly one token, the marking $M$ marks exactly $a$. Let $M'$ be the marking given by
$M \by{t} M'$. Since the S-component of $p_1$ contains a place in every cluster, no cluster different from $a_1$ can be marked
at $M'$. Symmetrically, no cluster different from $a_2$ can be enabled
at $M'$. So $M'$ does not mark any cluster, contradicting that $\W$ is sound.

Therefore we have that $t^\bullet \subseteq a'\cap P$ for every cluster $a \neq [o]$ and every transition $t$ of $a\cap T$. To show equality, again assume the contrary. Then with the same reasoning as above, after an occurrence sequence that ends with $a$, only places in $a\cap T$ are marked but $a\cap T$ is not marked. Thus the marking does not mark any cluster, again contradicting soundness.
\end{proof}

\begin{theorem}
\label{irredtheo}
Let $\W$ be an irreducible sound acyclic FC WF net. Then $\W$ contains only two clusters $[i],[o]$.
\end{theorem}

\begin{proofsketch}
Assume $\W$ contains more than two clusters. For every cluster $a \neq [o]$,
let $l(a)$ be the length of the longest path from $a$ to $o$ in the graph of $\W$. Let
$a_{\min}$ be any cluster such that $l(a_{\min})$ is minimal, and let $t$ be an arbitrary transition of $a_{\min}$ (notice that $a$ is cannot be $[i]$). By Lemma \ref{lemma2} there is a cluster
$a'$ such that $a'\cap T = t^\bullet$. If $a' \neq [o]$ then by
acyclicity we have $l(a') < l(a_{\min})$, contradicting the minimality of $a_{\min}$.
So we have $t^\bullet = [o]$ for every transition $t$ of $a_{\min}$.
If $a_{\min}$ has more than one transition, then the merge rule is applicable. Otherwise, since $\W$ is strongly connected some transition $t$must have an out-place in $a\cap P$ and thus by \ref{lemma2} $t^\bullet=a\cap P$. The the d-shortcut rule is applicable to $[t], t$ and $a$. 
In both cases we get a contradiction to irreducibility.
\end{proofsketch} 

The above lemmas prove that as long as the FC WF net $\W$ consists of more than two clusters and one transition, one of the rules is applicable. We now show that an application of the rules actually summarizes the net in polynomial time.

\begin{definition}
For every transition $t$, let ${\it shoc}(t)$ be the length of a longest maximal occurrence sequence containing $t$ minus $1$, and let ${\it Shoc}({\cal \W}) = \sum_{t \in T} {\it shoc}(t)$.
\end{definition}

Notice that if $\W$ has $K$ clusters then ${\it shoc}(t) \leq K-1$ holds for every transition $t$. Further, if $K=2$ then ${\it Shoc}(\W)= 0$.

\begin{theorem}
\label{thm:polcomp}
Every FC WF net $\W=(P,T,F,i,o)$ can be completely reduced by means of $|T|$ applications of the merge rule and ${\it Shoc}(\W)$ applications of the d-shortcut rule.
\end{theorem}

\begin{proof}
Let $\W$ and $\W'$ be FC WF nets such that $\W'$ is obtained from $\W$ by means of the merge or the d-shortcut rule. 
Assume that the merge rule is applied. Then we have $|T'| < |T|$ and ${\it Shoc}(\W') \leq {\it Shoc}(\W)$ because the rule reduces the number of transitions by one. 

Now assume that the d-shortcut rule is applied. Let $t$ be the removed transition, $c'$ the cluster that $t$ enabled unconditionally, $t'$ the single transition of $c'$ and $t'_s$ be the new transition produced by the rule. It holds that $|T'| = |T|$ because one transition was removed and one was added. Let $\sigma$ be any maximal occurrence sequence. 

If $\sigma$ does not contain $t$, it still exists in $\W'$. If $\sigma$ contains $t$, it also contains $t'$ because $t$ unconditionally enables $c'$ which has only a single transition. In $\W'$, the corresponding sequence to $\sigma$ is shorter by one as the occurrences of $t$ and $t'$ have been combined into a single occurrence of $t'_s$. Thus $\mathit{shoc}(t'_s) < \mathit{shoc}(t)$.

For all other transitions, $\mathit{shoc}(t'')$ either does not change (if one of the longest maximal occurrence sequences does not contain $t$) or decreases by one (otherwise). Thus in total, $\mathit{Shoc}(\W') <  \mathit{Shoc}(W)$.
\end{proof}

%\subsection{Cyclic stuff}

\begin{reflemma}{lem:soundCyclicSync}
Every sound cyclic FC WF net has a synchronizer.
\end{reflemma}

\begin{proof}
We begin by showing that every cyclic FC WF net has a loop.

Let $\pi$ be a cycle of the graph of the net ${\cal W}$. Let $t_1$ be an arbitrary transition occurring in $\pi$, and let $t_2$ be its successor in $\pi$. $t_1^\bullet \neq \{o\}$ because $o$ has no outgoing transitions, and hence no cycle contains $o$.

By soundness some reachable marking $M_1$ enables $t_1$. Furthermore it holds that $t^\bullet \cap {}^\bullet t_2 \neq \emptyset$. Let $M_1'$ be the marking reached after firing $t_1$ from $M_1$. 
Again by soundness, there is an occurrence sequence from $M_1$ that leads to the final marking. This sequence has to contain an occurrence of a transition of the cluster $[t_2]$ because there is a token on at least one place of this cluster. In particular, some prefix of this sequence leads to a marking $M_2$ that enables $t_2$. 

Repeating this argument arbitrarily for the transitions $t_1$, $t_2$, $t_3$, \ldots , $t_k = t_1$ of the cycle $\pi$, we conclude that there is an infinite occurrence sequence, containing infinitely many occurrences of transitions of the cycle $\pi$. Since the set of reachable markings is finite, this sequence contains a loop. 

Let now $\sigma$ be a minimal loop, i.e. the loop containing the least amount of transitions possible. We show that $\sigma$ contains a transition belonging to a synchronizer. For this proof, we will need to have a decomposition of $\W$ into S-components which corresponds to the idea of agents in a negotiation. Remember that the number of tokens in an S-component is constant.

Let $M$ be a marking where some transition $t_1$ in $\sigma$ is enabled. By soundness, there is some occurrence sequence $\tau$ enabled at $M$ which leads to the final marking. Let $t$ be the last transition in $\tau$ whose cluster has a transition $t^\ast$ that is contained in $\sigma$. We claim that $t^\ast$ is a synchronizer of $\sigma$.

Assume that this is not the case, i.e. there is a marking $M^\ast$ that enables $t^\ast$ (and since the net is free choice, also $t$) and also marks some other places $p \in (\bigcup_{t' \in\sigma, t' \neq t}{}^\bullet t')$. We first pick a subsequence $\tau' = t_1t_2t_3\ldots$ of $\tau$ such that $t_1=t$ and $t_i^\bullet \cap {}^\bullet t_{i+1} \neq emptyset$. Intuitively, we push one token towards the final marking. 

We now construct an occurrence sequence starting from the marking $M^\ast$ as follows: we start by firing $t$. Afterwards, whenever a transition in $\sigma$ is enabled, that transition is fired. If no such transition is enabled, but the next transition in $\tau'$ is enabled, that transition is fired. Otherwise, we add a minimal transition sequence that either ends with the final marking or enables some transition in $\sigma$ or $\tau'$.

First observe that there will always be some place $p \in {}^\bullet t_\sigma$ be marked for some $t_\sigma \in \sigma$. This is because initially one of those places is marked, and whenever a transition whose preset includes such a place is fired, it is a transition in $\sigma$. Thus the final marking can never be reached, because that marking only marks $o$ and no other place. Furthermore, since $\tau'$ is finite, we can only finitely often add transition sequences that enable a transition in $\tau'$ and subsequently fire it. Thus it has to be the case that transitions of $\sigma$ appear infinitely often. (Because of soundness and because the final marking is not reached, the sequence has to be infinite, otherwise there would be a deadlock.)

However, $t$ cannot be enabled ever again: the tokens of the S-components containing $t$ are somewhere along their path towards $o$ (because they have followed $\tau'$). Furthermore, the sequences we add until some transition in $\sigma$ is enabled increase the number of S-components whose token is on some place where a transition of $\sigma$ begins. This also can only happen finitely often.

Thus there must be a loop which uses only a subset of transitions in $\sigma$ contradicting the minimality of $\sigma$.

\end{proof}

\begin{reflemma}{lem:syncOnly}
For sound FC WF nets, a fragment of a minimal synchronizer can be reduced to only contain synchronizers by means of the d-shortcut and merge rule.
\end{reflemma}

\begin{proof}
We construct an auxiliary net which will be acyclic.

Let $\W = (P,T,F,i,o)$ be a WF net, $t$ a minimal synchronizer, $F_t$ its fragment $S$ the set of synchronizers in $F_t$. Every cycle in $F_t$ must contain some transition in $S$ by minimality of $t$. We now describe our auxiliary net:\\
We take the fragment of $t$, add a copy of the places $[t]\cap P$ and also an additional start place $i'$ and end place $o'$.
We then redirect all transitions that put tokens on $[t]$ to instead put them on the copy of $[t]\cap P$.
We add a transition that takes a token from $i'$ and places one on each place of $[t]\cap P$. Finally, we remove all synchronizers except those in $[t]$ and instead add a transition for each cluster $c$ from which we removed a synchronizer. The newly added transition has as preset the places $c\cap P$ of the cluster and as postset the new output place $o'$.

To ensure soundness, we remove all non-reachable clusters. Clusters could be unreachable because they do not appear on a path from $[t]$ to another synchronizer's cluster, but e.g. on a path from another synchronizer's cluster to $[t]$.

This net contains all paths in the fragment leading from $[t]$ to any synchronizer in $S$. Since all cycles in the fragment contain a transition in $S$, the net is acyclic. It can therefore be reduced by means of the d-shortcut and merge rules.

We apply the same rules sequence (minus the applications that shortcut over one of our newly created transitions which end in $o'$) to the original negotiation. The result is that the transition $t$ is shortcut to directly enable some $t' \in S$. We repeat this for the other synchronizers in $S$ and thus reduce the fragment to synchronizers only. 
\end{proof}

\begin{reftheorem}{thm:runtimeCyclic}
Every sound FC-CWF net can be summarized in at most $\mathcal{O}(|C|^4\cdot |T|)$ shortcut rule applications and $\mathcal{O}(|C|^4+|C|^2\cdot|T|)$ merge rule applications where $C$ is the set of clusters of the net.
Any unsound FC-CWF net can be recognized as unsound in the same time.
\end{reftheorem}

\begin{proof}
We again use ${\it Shoc}({\cal \W})$ as in the proof of Theorem \ref{thm:polcomp}.

We start by bounding the time it takes to reduce the minimal fragment with clusters $C_F$ and transitions $T_F$ to a synchronizer only fragment. During this reduction, only the d-shortcut and merge rule are applied. We can thus bound the number of merge rule applications by $T_F$. Using the auxiliary net construction from
\ref{lem:syncOnly} we see that for each synchronizer, a net with $\mathcal{O}(C_F)$ clusters and $\mathcal{O}(T_F)$ transitions is reduced. The reduction takes at most $\mathcal{O}(|C_F| \cdot |T_F|)$ applications of the d-shortcut rule and takes place at most once for each cluster in $\W$.

Next we bound the time to reduce a synchronizer only fragment with clusters $C_F$. We order the clusters in an arbitrary order $c_1,c_2,\ldots$. Consider the vector $(n(c_1), n(c_2),\ldots)$ where $n(c)$ is the number of incoming backward transitions to cluster $c$. With every application of the shortcut rule as described above, the first non-zero entry of this vector decreases, and all zero entries before the decreasing entry stay zero. Each entry $n(v)$ is bounded (after the merge rule is applied) by the number of clusters. Thus the number of shortcut rule applications is bounded by $|C_F|^2$.

After each application of the shortcut rule to a backwards transition, we apply the merge rule exhaustively. Since the cluster $c$ in the definition of the shortcut rule has at most $C_F$ outgoing transitions inside the fragment (remember that in a synchronizer only fragment, a transition has exactly a cluster as pre- and post place), at most $C_F$ new edges may result from the shortcut rule application, thus the number of merge rule applications is bounded by $|C_F|^3$.

Finally, applying the reduction procedure for acyclic net $\W_A$ with $C_A$ clusters and $T_A$ transitions takes again $|T_A|$ applications of the merge rule and $\mathit{Shoc}(\W_A)$ applications of the d-shortcut rule by Theorem \ref{thm:polcomp}. We can bound $\mathit{Shoc}(\W_A)$ by $\mathcal{O}(|C_A|\cdot |T_A|)$ by the definition of $\mathit{Shoc}$.

While the synchronizer only fragment is reduced, we may sometimes shortcut with a cluster which also has transitions leading outside of the fragment. This results in new transitions being produced that lead outside of the fragment. However, the total number of unique transitions produced in this way over during the whole reduction procedure is bounded by $|C|\cdot|T|$ where $C$ is the set of all clusters of the net: each such transition must originate from some cluster, and since we reduce a synchronizer only fragment, the new transition's post places are exactly the post places of an existing transition. Each of those transitions may be produced up to $C$ times and thus there are at most $|C|^2\cdot|T|$ merge rule applications.

We bound the number of clusters $|C_F|$ in each fragment and $|C_A|$ in the final acyclic net by $|C|$, the number of clusters in the original net. We further bound the number of transitions $|T_F|$ and $|T_A|$ by $|C|\cdot|T|$ by the above observation.

Summing up, we obtain for the shortcut rule $\mathcal{O}(|C| \cdot (|C|^2 + |C|^2\cdot|T|) + |C|^2 \cdot |T|) = \mathcal{O}(|C|^3 \cdot |T|)$ shortcut rule applications and $\mathcal{O}(|C| \cdot (|C|^3+|C|\cdot|T|)+|C|\cdot |T| + |C|^2\cdot |T|) = \mathcal{O}(|C|^4+|C|^2\cdot|T|)$ merge rule applications suffice.

\end{proof}

\subsection{Extension to generalized soundness}

We recall some definitions and theorems from \cite{desel2005free}

\begin{definition}[Trap]
A set $R$ of places of a Petri net is a trap if $R^\bullet \subseteq {}^\bullet R$	
\end{definition}

\begin{definition}[Home marking]
Let $N$ be a Petri net, $M_0$ an initial marking. A marking $M$ is a home marking of $(N,M_0)$ if it is reachable from every marking reachable from $M_0$.
\end{definition}

\begin{reftheorem}{thm:kSound}
Let $\W$ be a free choice workflow net. The following statements are equivalent:
\begin{itemize}
\item[(1)] $W$ is sound.
\item[(2)] $W$ is $k$-sound for some $k \geq 1$.
\item[(3)] $W$ is generalized sound.
\end{itemize}
\end{reftheorem}

\begin{proof}
(1) $\Rightarrow$ (2). Assume $W$ is sound. Fix $k \geq 1$, and let
$i^k \by{\sigma} M$ be an arbitrary occurrence sequence of $\W$.
We prove that there exists an occurrence sequence $\tau$ such that
$M \by{\tau} o^k$.

Consider the Petri net $\N$ obtained by adding to 
$\W$ a transition $t^*$ with $o$ as input and $i$ as output place. 
Since $W$ is sound, the Petri net $(\N, i)$ (i.e., the net $\N$ with
initial marking $i$) is live and bounded \cite{DBLP:journals/jcsc/Aalst98}.  
Moreover, the marking $i$ marks all proper traps of $\N$: indeed, by soundness there is an occurrence sequence $i \by{\sigma} i$ such that $\sigma$ contains all transitions of $\N$; since after such a sequence all traps are necessarily marked, the claim is proved. 

We now show that the net $(\N, i^k)$ as initial marking also satisfies these two properties. First, by Theorem 4.21, Theorem 4.27, and Theorem 5.8 of \cite{desel2005free}, 
adding tokens to a live and bounded marking of a free choice net preserves liveness and boundedness (see also exercise 4.8 of \cite{desel2005free}. So $(\N,i^k)$ is also live and bounded. Second, if $i$ marks all proper traps of $\N$, then obviously so does $i^k$.

By the Home Marking Theorem (\cite{desel2005free}, Theorem 8.11), the marking $i^k$ is a home marking of $\N$. Therefore, there exists an occurrence sequence $\tau'$ such that $M \by{\tau'} i^k$. Assume that $M(i) = k'$ for some $0 \leq k' \leq k$. Let $\tau$ be the result of removing from $\tau'$ the last $k-k'$ occurrences of transition $t^*$. We have $M \by{\tau} o^k$, and we are done.

(2) $\Rightarrow$ (1). Assume $W$ is $k$-sound for some $k \geq 1$. By the definition of $k$-soundness, the net $(\N, i^k)$ is bounded and deadlock-free. By Theorem 4.31 of
\cite{desel2005free}, $(\N, i^k)$ is live. By Commoner's theorem, 
$(\N, i)$ is also live, and by Theorem 5.8 of \cite{desel2005free} it is also bounded. 
Therefor, the workflow net $W$ is sound (see \cite{DBLP:journals/jcsc/Aalst98}).

(3) $\Rightarrow$ (2). Obvious from the definition.

(2) $\Rightarrow$ (3). In the statement (1) $\Rightarrow$ (2) we have actually shown that if $W$ is sound, then $W$ is sound for every $k \geq 1$. So assume that 
$W$ is $k$-sound for some $k \geq 1$. Then $W$ is sound, and therefore $W$ is $k$-sound for every $k \geq 1$. 
\end{proof}

\end{document}